\documentclass[11pt,letterpaper]{article}
\usepackage{iftex}

\ifPDFTeX
\usepackage[utf8]{inputenc}
\usepackage[T1]{fontenc}
\fi

\usepackage{lmodern}
\usepackage{amsmath}
\usepackage{amssymb}
\usepackage{amsthm}
\usepackage{mathtools}
\usepackage{thmtools}
\usepackage{array}
\usepackage{multirow}
\usepackage{enumitem}
\usepackage{thmtools} \usepackage{thm-restate}

\usepackage[margin=1in]{geometry}

\usepackage{microtype}

\usepackage{xcolor}

\usepackage[algo2e,ruled,vlined,linesnumbered]{algorithm2e}

\usepackage{dsfont}
\usepackage{comment}
\usepackage{lineno}

\usepackage[
	style=alphabetic,
    maxalphanames=4,
    minalphanames=4,
    maxcitenames=4,
    minnames=1,
    maxbibnames=20,
	backref=true,
	doi=true,
	url=true,
	backend=biber
]{biblatex}

\usepackage[ocgcolorlinks]{hyperref} 
\usepackage{cleveref}

\allowdisplaybreaks[1]

\makeatletter
\g@addto@macro\bfseries{\boldmath}
\makeatother

\makeatletter
\g@addto@macro\mdseries{\unboldmath}
\g@addto@macro\normalfont{\unboldmath}
\g@addto@macro\rmfamily{\unboldmath}
\g@addto@macro\upshape{\unboldmath}
\makeatother

\DeclareCiteCommand{\citem}
    {}
    {\mkbibbrackets{\bibhyperref{\usebibmacro{postnote}}}}
    {\multicitedelim}
    {}

\renewcommand*{\multicitedelim}{\addcomma\space}

\AtEveryCitekey{\clearfield{note}}

\makeatletter
\AtBeginDocument{%
    \newlength{\temp@x}%
    \newlength{\temp@y}%
    \newlength{\temp@w}%
    \newlength{\temp@h}%
    \def\my@coords#1#2#3#4{%
      \setlength{\temp@x}{#1}%
      \setlength{\temp@y}{#2}%
      \setlength{\temp@w}{#3}%
      \setlength{\temp@h}{#4}%
      \adjustlengths{}%
      \my@pdfliteral{\strip@pt\temp@x\space\strip@pt\temp@y\space\strip@pt\temp@w\space\strip@pt\temp@h\space re}}%
    \ifpdf
      \typeout{In PDF mode}%
      \def\my@pdfliteral#1{\pdfliteral page{#1}}
      \def\adjustlengths{}%
    \fi
    \ifxetex
      \def\my@pdfliteral #1{}
      \def\adjustlengths{\setlength{\temp@h}{-\temp@h}\addtolength{\temp@y}{1in}\addtolength{\temp@x}{-1in}}%
    \fi%
    \def\Hy@colorlink#1{%
      \begingroup
        \ifHy@ocgcolorlinks
          \def\Hy@ocgcolor{#1}%
          \my@pdfliteral{q}%
          \my@pdfliteral{7 Tr}
        \else
          \HyColor@UseColor#1%
        \fi
    }%
    \def\Hy@endcolorlink{%
      \ifHy@ocgcolorlinks%
        \my@pdfliteral{/OC/OCPrint BDC}%
        \my@coords{0pt}{0pt}{\pdfpagewidth}{\pdfpageheight}%
        \my@pdfliteral{F}
        \my@pdfliteral{EMC/OC/OCView BDC}%
        \begingroup%
          \expandafter\HyColor@UseColor\Hy@ocgcolor%
          \my@coords{0pt}{0pt}{\pdfpagewidth}{\pdfpageheight}%
          \my@pdfliteral{F}
        \endgroup%
        \my@pdfliteral{EMC}%
        \my@pdfliteral{0 Tr}
        \my@pdfliteral{Q}%
      \fi
      \endgroup
    }%
}
\makeatother

\colorlet{DarkRed}{red!50!black}
\colorlet{DarkGreen}{green!50!black}
\colorlet{DarkBlue}{blue!50!black}

\hypersetup{
	linkcolor = DarkRed,
	citecolor = DarkGreen,
	urlcolor = DarkBlue,
	bookmarksnumbered = true,
	linktocpage = true
}

\allowdisplaybreaks[3]

\DontPrintSemicolon
\SetKwInOut{Input}{Input}\SetKwInOut{Output}{Output}

\SetCommentSty{mycommfont}

\declaretheorem[numberwithin=section]{theorem}
\declaretheorem[numberlike=theorem]{lemma}

\declaretheorem[numberlike=theorem]{corollary}
\declaretheorem[numberlike=theorem]{definition}

\usepackage{pgffor} 
\foreach \x in {A,...,Z}{%
	\expandafter\xdef\csname m\x\endcsname{\noexpand\mathbf{\x}}
}

\newcommand{\F}{\mathbb{F}}

\newcommand{\Z}{\mathbb{Z}}

\newcommand{\init}{\text{(init)}}
\newcommand{\dyn}{\text{(dyn)}}

\DeclareMathOperator{\poly}{poly}

\usepackage[textwidth=2cm,textsize=footnotesize]{todonotes}

\title{Fast Deterministic Fully Dynamic Distance Approximation}

\usepackage{authblk}

\author[1]{Jan van den Brand}
\affil[1]{Simons Institute \& UC Berkeley, USA}
\author[2]{Sebastian Forster}
\author[2]{Yasamin Nazari}
\affil[2]{University of Salzburg, Austria}

\date{}

\makeatletter
\renewcommand{\paragraph}{%
	\@startsection{paragraph}{4}%
	{\z@}{1.25ex \@plus 1ex \@minus .2ex}{-1em}%
	{\normalfont\normalsize\bfseries}%
}
\makeatother

\begin{document}
\maketitle
\pagenumbering{roman}

\begin{abstract}
    In this paper, we develop deterministic fully dynamic algorithms for computing approximate distances in a graph with worst-case update time guarantees.
In particular, we obtain improved dynamic algorithms that, given an unweighted and undirected graph $G=(V,E)$ undergoing edge insertions and deletions, and a parameter $ 0 < \epsilon \leq 1 $, maintain $(1+\epsilon)$-approximations of the $st$-distance between a given pair of nodes $ s $ and $ t $, the distances from a single source to all nodes (``SSSP''), the distances from multiple sources to all nodes (``MSSP''), or the distances between all nodes (``APSP'').

Our main result is a deterministic algorithm for maintaining $(1+\epsilon)$-approximate $st$-distance with worst-case update time $O(n^{1.407})$ (for the current best known bound on the matrix multiplication exponent $\omega$).
This even improves upon the fastest known randomized algorithm for this problem. Similar to several other well-studied dynamic problems whose state-of-the-art worst-case update time is $O(n^{1.407})$, this matches a conditional lower bound~\citem[BNS, FOCS 2019]{BrandNS19}.
We further give a deterministic algorithm for maintaining $(1+\epsilon)$-approximate single-source distances with worst-case update time $O(n^{1.529})$, which also matches a conditional lower bound.

At the core, our approach is to combine algebraic distance maintenance data structures with near-additive emulator constructions.
This also leads to novel dynamic algorithms for maintaining $(1+\epsilon, \beta)$-emulators that improve upon the state of the art, which might be of independent interest.
Our techniques also lead to improved randomized algorithms for several problems such as exact $st$-distances and diameter approximation.
\end{abstract}

\vspace{7ex}

\noindent

\newpage
\tableofcontents
\newpage
\pagenumbering{arabic}


\section{Introduction}

From the procedural point of view, an algorithm is a set of instructions that outputs the result of a computational task for a given input.
This static viewpoint neglects that computation is often not a one-time task with input data in successive runs of the algorithm being very similar. 
The idea of dynamic graph algorithms is to explicitly model the situation that the input is constantly undergoing changes and the algorithm needs to adapt its output after each change to the input.
This paradigm has been highly successfully applied to the domain of graph algorithms. The major goal in designing dynamic graph algorithms is to spend as little computation time as possible for processing each update to the input graph.

Despite the progress on dynamic graph algorithms in recent years, many state-of-the-art solutions suffer from at least one of the following restrictions:
(1) Many dynamic algorithms only support one type of updates, i.e., are incremental (supporting only insertions) or decremental (supporting only deletions).
\emph{Fully dynamic} algorithms support both types of updates.
(2) Many dynamic algorithms only achieve amortized update time guarantees, i.e., the stated bound only holds ``on average'' over a sequence of updates with individual updates possibly taking significantly more time than the stated amortized bound.
{Worst-case} bounds also hold for individual updates, which for example is relevant in real-time systems.
(3) Many dynamic algorithms are randomized.
(i) On one hand, this means these algorithms only give probabilistic guarantees on correctness or running time that do not hold in all cases.
(ii) On the other hand, randomized algorithms often do not allow the ``adversary'' creating the sequence of updates to be \emph{adaptive} in the sense that it may react to the outputs of the algorithm\footnote{This type of adversary is called ``adaptive online adversary'' in the context of online algorithms~\cite{Ben-DavidBKTW94}. Note that despite being allowed to choose the next update in its sequence based on the outputs of the algorithm so far, this adversary may not explicitly observe the internal random choices of the algorithm.}. This is because the power of randomization can in many cases only be unleashed if the adversary is \emph{oblivious} to the outputs of the algorithm, which guarantees probabilistic independence of the random choices made by the algorithm.
\emph{Deterministic} algorithms avoid these two issues.

While these restrictions are not prohibitive in certain settings, they obstruct the general-purpose usage of dynamic algorithms as ``black boxes''.
Thus, the ``gold standard'' in the design of dynamic algorithms should be deterministic fully dynamic algorithms with worst-case update time bounds.
To date, there is only a limited number of problems that admit such algorithms and additionally have time bounds that match (conditional) lower bounds (say up to subpolynomial factors).
To the best of our knowledge, this is the case only for $ (2 + \epsilon) $-approximate maximum fractional matching and minimum vertex cover~\cite{BhattacharyaHN17}, $ (2 \Delta - 1) $-edge coloring~\cite{BhattacharyaCHN18}, $ (1 - \epsilon) $-approximate densest subgraph~\cite{SawlaniW20}, connectivity~\cite{ChuzhoyGLNPS20}, minimum spanning tree~\cite{ChuzhoyGLNPS20}, and edge connectivity~\cite{JinS21}.

In this paper, we add two important problems to this list: $ (1 + \epsilon) $-approximate $st$ distances and $ (1 + \epsilon) $-approximate single-source distances in unweighted, undirected graphs.
For current bounds on the matrix-multiplication exponent\footnote{
Two $n\times n$ matrices can be multiplied in $O(n^\omega)$ operations with $\omega \le 2.373$ \cite{Williams12,Gall14,AlmanW21}. 
We write $O(n^{\omega(a,b,c)})$ for the complexity of multiplying an $n^a \times n^b$ by $n^b \times n^c$ matrix \cite{GallU18}.} $ \omega$,
our deterministic worst-case update times for these problems are $ O (n^{1.407}) $ and $ O (n^{1.529}) $, respectively,
and match conditional lower bounds from~\cite{BrandNS19} up to subpolynomial factors.
In particular, the dynamic $ (1 + \epsilon) $-approximate $st$ distance currently shares this conditional lower bound and the upper bound we derive with an array of other dynamic problems such as $st$ reachability and cycle detection in directed graphs, maximum matching size, or determinant and rank of a matrix \cite{Sankowski04,Sankowski05,Sankowski07,BrandNS19}.

Apart from our main results for $st$ and single-source distances, we also obtain novel results for approximating multi-source distances, all-pairs distances, and the diameter; see \Cref{sec:results} for a detailed overview on our results.

Summarized in one sentence, our results are obtained by combining algebraic bounded-distance data structures with near-additive emulator constructions (see Definition \ref{def:spanners_emulators}) and then obtaining distance estimates from such an emulator.
A similar strategy was employed by recent related work of~\cite{BHGWW2021}.  One major ingredient of their approach is equipping the algebraic distance data structure of Sankowski~\cite{Sankowski04,Sankowski05} with a path-reporting mechanism similar to Seidel's technique for APSP in the static setting~\cite{Seidel95}.
This allows them to maintain a near-additive spanner -- which fits their path-reporting purposes -- but together with other parts of their algorithm introduces randomization. By using certain types of emulators instead of spanners, we can obtain a faster, deterministic algorithm. In particular, we can tailor the algebraic data structures better to our needs due to several nice properties of our emulators, for example that their structure changes \textit{slowly} and \textit{locally}.

In the remainder of this section we state all our results and compare them with related work. In Section~\ref{sec:overview}, we give an overview of our main ideas and technical contributions.
We subsequently provide the full details: 
Section~\ref{sec:emulator} focuses on the combinatorial aspects of
maintaining ``mid-sparsity'' emulators and its consequences to $st$ and single source distances. Section~\ref{sec:sparse emulator} gives a more general emulator result by further sparsifying this ``mid-sparsity'' emulators which leads to further applications for MSSP and APSP distances. Finally, in Section~\ref{sec:algebraic} we design an algebraic bounded-distance distance data structure used in our emulator constructions.
In Appendix~\ref{app:randomized} we explain how almost immediate consequences of our primitives lead to improved bounds for diameter approximation and APSP distance oracles.

\subsection{Our Results and Comparison with Related Work}\label{sec:results}

In this section, we summarize our main results for deterministic fully dynamic distance computation ($st$, SSSP, APSP, and MSSP supporting distance queries) and emulators. 
A summary of our algorithms for maintaining $(1+\epsilon)$-approximate distances with their worst-case update time guarantees can be found in Table \ref{tab:results}. 
In addition to these deterministic results, our techniques also give improved randomized solutions for diameter approximation, and subquadratic update-time $(1+\epsilon)$-APSP distance oracles\footnote{%
By a ``distance oracle'', we mean a data structure that supports fast queries. 
Our goal -- unlike many static algorithms -- is not optimizing the space of this data structure.} 
with sublinear query time. We next discuss each of these results and compare them with related work. 
Throughout this paper we assume that we are given an unweighted graph with $n$ nodes and $m$ edges. 

\begin{table}[t] 
    \centering
    \def\arraystretch{1.5} 
\scalebox{0.8}{    
    \begin{tabular}[t]{|c|c|c|c|c|}
\hline 
Approx & Type & \multicolumn{2}{c|}{Worst-case update} & Reference \tabularnewline
\hline
\hline
$1+\epsilon$ & st & \multicolumn{2}{c|}{$O(n^{1.407}\epsilon^{-2} \log \epsilon^{-1})$} & \Cref{thm:st}
\tabularnewline
\hline 
$1+\epsilon$ & SSSP & \multicolumn{2}{c|}{$O(n^{1.529}\epsilon^{-2} \log \epsilon^{-1} )$} &  \Cref{thm:sssp}
\tabularnewline
\hline
$1+\epsilon$ & $k$-MSSP & \multicolumn{2}{c|}{$O(n^{1.529}+kn) \cdot O({\epsilon}^{-1})^{\sqrt{2\log_{1/\epsilon} n}}$} & \Cref{thm:MSSP}
\tabularnewline
\hline 
$1+\epsilon$ & APSP & \multicolumn{2}{c|}{$O (n^2) \cdot  O({\epsilon}^{-1})^{\sqrt{2\log_{1/\epsilon} n}}$} & \Cref{cor:APSP}
\tabularnewline
\hline 
$(1+\epsilon, n^{o(1)})$ & Emulators & \multicolumn{2}{c|}{$O(n^{1.407}\epsilon^{-2} \log \epsilon^{-1})$} & \Cref{lem:sparse_det_emulator}
\tabularnewline
\hline 
\end{tabular}
}
    \caption{\footnotesize Summary of our deterministic results for distance and emulators (our randomized ones are not included). 
    By $k$-MSSP we mean multi-source distances from $k$ sources. For the exact dependence on $\omega$, see the respective theorems.
    }
    \label{tab:results}
\end{table}

\paragraph{Deterministic $(1+\epsilon)$-$st$ distances.}
Our main result is a deterministic, fully dynamic algorithm for maintaining a $(1 + \epsilon)$-approximation of the distance between a fixed pair of nodes $s,t \in V$ whose worst-case update time matches a conditional lower bound.

\begin{restatable}{theorem}{st}\label{thm:st}
Given an unweighted undirected graph $G=(V,E)$ and a pair of nodes $s$ and $t$, there is a fully-dynamic data structure for maintaining $(1+\epsilon)$-distances between $s$ and $t$ deterministically with
\begin{itemize}[nosep]
    \item Preprocessing time of $O(n^{\omega}\epsilon ^{-2}\log\epsilon^{-1})$, where $\omega \le 2.373$.
    \item Worst-case update time of $O((n^{\omega(1,1,\mu)-\mu} + n^{\omega(1,\mu,\nu)-\nu} + n^{\mu+\nu} + n^{4/3}){\epsilon^{-2}}\log\epsilon^{-1})$
    for any parameters $0\le\nu\le\mu\le1$, which is $O(n^{1.407} \epsilon^{-2} \log \epsilon^{-1})$ for current $\omega$
    ($\mu\approx 0.856$, 
    $\nu \approx 0.551$).
\end{itemize}
\end{restatable}

We are not aware of any non-trivial deterministic algorithms with worst-case update time for maintaining the exact or $(1+\epsilon)$-approximate $st$-distance under both insertions and deletions.\footnote{In independent work, \cite{KarczmarzMS22} obtained such deterministic bounds for (the more general) directed graphs when restricting to only edge insertions. 
}

When relaxing determinism to randomization against adaptive adversaries, the previously fastest fully-dynamic algorithm for $st$-distance has worst-case update time $O(n^{1.724})$ \cite{Sankowski05,BrandNS19} and maintains the distance exactly for unweighted directed graphs. We later also show that if randomization is allowed, our approach also improves the bound for \textit{exact} $st$-distances to $O(n^{1.7035})$.

The previously fastest fully dynamic for unweighted, undirected graphs is implied by the approach of \cite{BHGWW2021} and yields a worst-case update time of $ O (n^{1.529}) $; this algorithm employs randomization against oblivious adversaries, and in addition to the approximate distance can also report an $st$-path of the corresponding length. Despite being deterministic, our algorithm improves upon these upper bounds.

Moreover, for \textit{current bounds} on $\omega$, our result closes the gap between previous upper bounds and a conditional lower bound for $(1+\epsilon)$-approximate dynamic $st$-distances on unweighted undirected graphs \cite{BrandNS19}.
This conditional lower bound is based on a hardness assumption called ``uMv-hinted uMv'' where a vector-matrix-vector product must be computed after receiving hints about the structure of the three inputs.
This assumption formalizes the current barrier for improving upon algorithms for various fully dynamic problems
such as directed $st$-reachability, maximum matching size, directed cycle detection, directed $k$-cycle and $k$-path detection, and on the algebraic side, maintaining determinant and rank of a dynamic matrix. All of these problems admit an $O(\min_{0\le\nu\le\mu\le1} (n^{\omega(1,1,\mu)-\mu} +n^{\omega(1,\mu,\nu)-\nu} + n^{\mu+\nu})) $ worst-case update time (which for current $ \omega $ amounts to $ O(n^{1.407}) $) and -- assuming hardness of ``uMv-hinted uMv'' -- no dynamic algorithm for these problems can improve upon this by a polynomial factor.
While the nature of conditional lower bounds can never rule out the existence of faster algorithms with certainty,  we believe that these connections provide evidence that a substantial breakthrough will be necessary in order to improve upon the update time of our algorithm.
We further note that closing the update-time gap between fully dynamic $st$-reachability and $st$-distance was raised as an important open problem by Sankowski~\cite{Sankowski08}; our bound for $(1+\epsilon)$-approximate $st$-distance in undirected graphs partially resolves this question.

\paragraph{Deterministic $(1+\epsilon)$-SSSP.}

Our second result is a deterministic, fully dynamic algorithm for maintaining $(1+\epsilon)$-single source distances whose worst-case update time matches a conditional lower bound.

Formally we show the following. 
\begin{restatable}{theorem}{sssp}\label{thm:sssp}

Given an unweighted undirected graph $G=(V,E)$ and a single source $s$, and $0 <\epsilon <1$, there is a deterministic fully-dynamic data structure for maintaining $(1+\epsilon)$-distances from $s$ with
\begin{itemize}[nosep]
    \item Preprocessing time of $O(n^{\omega}\epsilon^{-2}\log \epsilon^{-1})$, where $\omega \le 2.373$.
    \item Worst-case update time of $O((n^{\omega(1,1,\mu)-\mu}+n^{1+\mu}) \epsilon^{-2}\log \epsilon^{-1})$ for any $0\le\mu\le1$. For current bounds on $\omega$ and the best choice of $\mu\approx 0.529$, this is $O(n^{1.529} \epsilon^{-2})$.
\end{itemize}
\end{restatable}
As with the $st$ case, we are not aware of any non-trivial deterministic algorithms with worst-case update time for maintaining exact or $(1+\epsilon)$-approximate SSSP under both insertions and deletions.
When relaxing determinism to randomization against adaptive adversaries, the previously fastest fully-dynamic algorithm for $ (1 + \epsilon) $-approximate SSSP in unweighted graphs has a much slower worst-case update time of $O(n^{1.823})$~\cite{BrandN19} 
(albeit that bound also holds for directed weighted graphs).

The fastest fully dynamic algorithm for unweighted, undirected graphs is implied by the approach of \cite{BHGWW2021} and yields a worst-case update time of $ O (n^{1.529}) $; this algorithm employs randomization against oblivious adversaries, and in addition to the approximate distance can also report an $st$-path of the corresponding length.
Our result matches the update time of the distance maintenance problem with a deterministic algorithm and additionally improves the dependence on the error parameter $ \epsilon $ from $ (1/\epsilon)^{O(\sqrt{\log_{1/\epsilon} n})} $ to a small polynomial.
Moreover, this update time matches a conditional lower bound stated in \cite{BrandNS19} based on the hardness assumption ``Mv-hinted Mv''.
This is a similar type of hardness assumption as discusses in the $st$-case, but tuned to single source problems. We emphasize again that our approximate $st$ result matches the conditional lower for \textit{current} $\omega$, whereas our approximate SSSP bound matches the conditional lower bound for any $\omega$. 

\paragraph{Deterministic Sparse Emulators.}

The main tool developed and applied in this paper is a novel fully dynamic algorithm for maintaining $(1 + \epsilon, \beta)$-emulators with various trade-offs, which might be of independent interest.

\begin{definition}\label{def:spanners_emulators}
Given a graph $G = (V,E)$, an \emph{$(\alpha, \beta)$-emulator} of $G$ is a graph $H=(V, E')$ (that is not necessarily a subgraph of $G$ and might be weighted) in which $d_G(u,v) \leq d_H(u,v) \leq \alpha \cdot d_G(u,v) +\beta$ for all pairs of nodes $ u, v \in V $.
If $ H $ is a subgraph of $ G $, then $ H $ is an \emph{$(\alpha, \beta)$-spanner} of $G$. 
\end{definition}

In this paper, we are mainly interested in so-called \emph{near-additive} emulators and spanners, as introduced by~\cite{ElkinP04}, for which $ \alpha = 1 + \epsilon $ for any parameter $ \epsilon > 0 $ and $ \beta $ is a function of $ \epsilon $.
A influential construction of Thorup and Zwick~\cite{TZ2006emulators} gives $(1 + \epsilon, \beta)$-spanners of size $\tilde{O}(n^{1+1/k})$\footnote{%
Throughout this paper, we use $ \tilde O (\cdot) $-notation to suppress terms that are polylogarithmic in $ n $, the number of nodes of the graph.}
and with $\beta= O(1/\epsilon)^{k}$ for any $0<\epsilon\leq1$ and $2 \leq k \leq \log n$ that can statically be computed in time $\tilde{O}(mn^{1/k})$.\footnote{In static settings there are somewhat more involved algorithms for near-additive emulators that lead to slightly better tradeoffs in specific parameter settings (e.g.~see \cite{EN2018}).}
For any constant $ \epsilon $, this allows for a $(1 + \epsilon, n^{o(1})$-spanner of size $ n^{1+o(1)} $.

In this paper, we obtain the following result for maintaining near-additive emulators.
\begin{restatable}{lemma}{emulator}\label{lem:sparse_det_emulator}
Given an unweighted, undirected graph $G=(V,E)$, 
parameters $0<\epsilon<1$ and $2 \leq k \leq \log n$, 
we can maintain a $(1+\epsilon, \beta)$-emulator of $G$ with size $\tilde{O}(n^{1+1/k})$, 
where $\beta= O(1/\epsilon)^{k}$ deterministically with worst-case update time of $\max(\tilde{O}(n^{4/3+1/k}), O(n^{1.407}\epsilon^{-2}\log \epsilon^{-1} ))$. Here the latter term of the update time has the same dependence on $\omega$ as \Cref{thm:st}.
The preprocessing time of this algorithm is $O(n^{\omega} \epsilon^{-2}\log \epsilon^{-1})$.
\end{restatable}

This result should mainly be compared to the fully dynamic algorithm of~\cite{BHGWW2021} for maintaining a $(1+\epsilon, n^{o(1)})$-spanner of size $ n^{1+o(1)} $ with worst-case update time $O(n^{1.529})$ for any constant $ \epsilon > 0 $ that employs randomization against an oblivious adversary.
We improve upon the result of~\cite{BHGWW2021} both in running time and by having a deterministic algorithm at the cost of maintaining emulators instead of spanners.
Other works on maintaining spanners or emulators give a multiplicative stretch $ \alpha \geq 3 $~\cite{AusielloFI06,Elkin11,BaswanaKS12,BodwinK16,BernsteinFH21,ForsterG19,BernsteinBGNSSS20} or are restricted to a partially dynamic setting~\cite{BernsteinR11,HKN2013}.

\paragraph{Deterministic $(1+\epsilon)$-MSSP.}

Another implication of our techniques is an algorithm for $(1+\epsilon)$-multi-source distances. In Section \ref{sec:MSSP}, we give an algorithm that combines our sparse emulator construction with the algebraic techniques to prove the following theorem.

\begin{restatable}{theorem}{mssp} \label{thm:MSSP}
Given an unweighted, undirected graph $G=(V,E)$, and $0<\epsilon<1$, and a fixed set of sources $S$, we can maintain $(1+\epsilon)$-approximate distances from $S$ (i.e.~pairs in $S \times V$) deterministically with
     $O((n^{\omega(1,1,\mu)-\mu}+n^{1+\mu}+ |S| \cdot n) \cdot O(1/\epsilon)^{\sqrt{2\log_{1/\epsilon} n}}$ worst-case update time, which for current $\omega$ is $O(n^{1.529}+ |S| \cdot n) \cdot O(1/\epsilon)^{\sqrt{2\log_{1/\epsilon} n}}$. The preprocessing time is $O(n^{\omega}) \cdot O(\tfrac{1}{\epsilon})^{\sqrt{2\log_{1/\epsilon} n}}$.
\end{restatable}

Hence we can maintain distances from up to $\tilde{O}(n^{0.52})$ sources in almost (up to an $n^{o(1)}$ factor) the same time as maintaining distances from a single-source.

\paragraph{Deterministic $(1+\epsilon)$-APSP.}
One implication of \Cref{thm:MSSP} (by simply setting $S=V$) is a deterministic fully-dynamic algorithm for maintaining all-pairs-shortest path that nearly (up to an $n^{o(1)}$ factor) matches the trivial lower bound of $\Omega(n^2)$ time per update for this problem. More formally,
\begin{restatable}{corollary}{apsp}\label{cor:APSP}
Given an unweighted, undirected graph $G=(V,E)$, and $0<\epsilon<1$, we can maintain $(1+\epsilon)$-all-pairs distances deterministically with $O(n^{2}) \cdot O(\tfrac{1}{\epsilon})^{\sqrt{2\log_{1/\epsilon} n}}$ worst-case update time. The preprocessing time is $O(n^{\omega}) \cdot O(\tfrac{1}{\epsilon})^{\sqrt{2\log_{1/\epsilon} n}}$.
\end{restatable}

It is worth mentioning that there is another (simpler) approach to obtain this bound that we will discuss in Section \ref{sec:APSP}.
The previous comparable bounds for this problem either used randomization \cite{BrandN19} or have amortized bounds \cite{DemetrescuI04,Thorup04}.
The fastest deterministic algorithm with worst-case guarantee that maintains exact shortest paths unweighted, directed graphs and has an update time of $ \tilde O (n^{2.6})$~\cite{GutenbergW20}.

We will next show several result that, unlike our previous bounds, are randomized. This includes an improved bound for exact $st$-distances using our new algebraic data structures (details in \Cref{app:exact_st}), and two other implications of our 
dynamic multi-source algorithms (details in \Cref{sec:appendix:randomized}).
\paragraph{Exact $st$-distances.}
Our new dynamic algorithm for maintaining bounded distances also leads to improved bounds for dynamic exact $st$-distances in \textit{directed graphs}, if we allow randomization. For current $\omega$, we obtain a worst-case update time of $O(n^{1.7035})$,
improving upon the previous best bound of $O(n^{1.7643})$ \cite{Sankowski05,BrandNS19}.
\begin{restatable}{theorem}{exactst}\label{thm:exactst}
For any $0\le\nu\le\mu\le1$ and $0\le h\le n$, there exists a randomized dynamic algorithm that maintains exact $st$-distances in directed graphs.
The preprocessing time is $\tilde O(hn^\omega)$ and the worst-case update time per edge insertion or deletion is $\tilde{O}(h(n^{\omega(1,1,\mu)-\mu}+n^{\omega(1,\mu,\nu)-\nu}+n^{\mu+\nu}+(n/h)^{2}))$.
After each update, the algorithm returns the exact $st$-distance and the result is correct with high probability.
The algorithm works against an adaptive adversary.

For current bounds on $\omega$, 
this is $O(n^{1.7035})$ time per update 
(with 
$\mu\approx 0.8556$, 
$\nu \approx 0.5512$ and 
$\log_n(h) \approx 0.2966$).
\end{restatable}

\paragraph{Randomized Approximate Diameter.}
We can maintain a nearly-$(3/2 +\epsilon)$-approximation
of the diameter in fully dynamic unweighted graphs in $O(n^{1.596})\cdot (\frac{1}{\epsilon})^{O(1)}$ worst-case update time against an adaptive adversary. See \Cref{cor:diam} for details.
This is done by using our emulator to compute $(1+\epsilon)$-MSSP algorithms for certain sets $S$ of size $\tilde{O}(\sqrt{n})$ based on an algorithm by \cite{RV13}. 

Previously, the fastest fully-dynamic algorithm with this approximation guarantee by \cite{BrandN19} had a worst-case update time of $O(n^{1.779})$ and employed randomization against an adaptive adversary. We get better bounds by combining our sparse emulator algorithms with the algorithm of \cite{BrandN19}. 

Dynamic diameter was also analyzed in the partially dynamic setting \cite{AnconaHRWW19,ChoudharyG20}, e.g.~there exists a nearly-$(3/2 + \epsilon)$-approximate decremental algorithm with $m^{1+o(1/\epsilon)}\sqrt{n}/\epsilon^2$ expected total update time \cite{AnconaHRWW19}.

\paragraph{$(1+\epsilon)$-APSP Distance Oracles with Sublinear Query.} Another implication of our new approach for dynamic $(1+\epsilon)$-MSSP is an improved bound for maintaining a data structure supporting all-pairs distance queries that has subquadratic update time $O(n^{1.788}) \cdot O (\tfrac{1}{\epsilon})^{\sqrt{2\log_{1/\epsilon} n}})$ and a small polynomial query time $O (n^{0.45} \epsilon^{-2})$ against an adaptive adversary.
See \Cref{cor:subquad_APSP} for details.

Our result directly improves upon the $O(n^{1.862}\epsilon^{-2}\log \epsilon^{-1})$ update time of a corresponding algorithm by~\cite{BrandN19} which has the same query time as ours and also employs randomization against an adaptive adversary.
The algorithm of~\cite{BrandN19} internally maintains $(1+\epsilon)$-approximate MSSP and thus our result is almost directly implied by our improvement for maintaining approximate MSSP.

\subsection{Further Related Work}
Several state-of-the art dynamic algorithms employ an algebraic approach (i.e.~use fast matrix multiplication) for maintaining reachability and distance information.
As a conditional lower bound by Abboud and Vassilevska Williams~\cite{AbboudW14} shows, this is inherent in certain regimes: 
Unless one is able to multiply two $n\times n$ boolean matrices in $ O(n^{3 - \delta})$ time for some constant $\delta>0$,
no fully dynamic algorithm for $st$ reachability in directed graphs can beat $ O(n^{2-\delta'}) $ update and query time and $ O(n^{3 - \delta'}) $ preprocessing time  (for some constant $ \delta'> 0 $).

While not explicitly stated in \cite{AbboudW14}, the same conditional lower bound extends to fully dynamic $(1 + \epsilon) $-approximate $st$ distances on undirected unweighted graphs for a small enough constant $ \epsilon $.

In the same spirit, \cite{BHNW2021} obtained a more refined conditional lower bound for combinatorial algorithms maintaining sparse near-additive spanners and emulators based on the Combinatorial $k$-Clique hypothesis.

The use of algebraic techniques for maintaining reachability and distance information can be traced back to the path counting approaches of King and Sagert~\cite{KingS02} and Demetrescu and Italiano~\cite{DemetrescuI00}.
Sankowski~\cite{Sankowski04} subsequently developed a more general framework for maintaining the adjoint of a matrix and applied it to maintaining reachability in directed graphs~\cite{Sankowski04} and distances in unweighted, directed graph~\cite{Sankowski05}.
This approach was further refined which led to improved dynamic algorithms for reachability \cite{BrandNS19} as well as for approximate distances~\cite{BrandN19}.
Recently, such algebraic data structures have been enriched to maintain ``witnesses'' that allow reporting paths in addition to the pure reachability/distance information:
the path reporting mechanism of~\cite{BHGWW2021} uses randomization against an oblivious adversary and the one of~\cite{KarczmarzMS22} uses randomization against an adaptive adversary.
The latter paper also contains deterministic bounds for incremental approximate shortest paths independently of our work.

\section{Technical Overview}\label{sec:overview}

In this section we give a high-level overview of our technical contributions. In \Cref{sec:overview:emulator}, we start by presenting deterministic algorithms for maintaining $(1+\epsilon,2)$ and $(1+\epsilon,4)$-emulators with applications respectively in $(1+\epsilon)$-SSSP and $(1+\epsilon)$-$st$ distances.
These emulator algorithms slightly extend a known ``localization''~\cite{HKN2013} of the (randomized) additive emulator construction~\cite{DorHZ00} and have two properties crucial for our bounds:
(1) They are based on a ``deterministic'' and ``slowly changing'' hitting set of high-degree neighborhoods. 
(2) For assigning the edge weights, we only need to compute bounded pairwise distances between the smaller set of nodes involving the hitting set.
We show that in our setting we can -- instead of using a standard randomized approach -- deterministically maintain an approximate solution to this particular hitting-set instance with low recourse.

In \Cref{sec:overview:algebraic}, we then design an algebraic data structure for maintaining bounded distances in such a way that it can deal with a gradually changing hitting set efficiently.
Following the approach by Sankowski \cite{Sankowski05}, maintaining small distances between certain vertices reduces to maintaining a submatrix of some dynamic matrix inverse.
We modify the dynamic matrix inverse algorithm of \cite{BrandNS19} to efficiently maintain such a submatrix. 
In general, the algorithm of \cite{BrandNS19} has faster update but slower query time compared to other dynamic matrix inverse algorithms \cite{Sankowski04}.
However, by exploiting that the queries will be located within some specified submatrix, we can speed up the query complexity.
Using this additional information about the location of the queries, we can periodically precompute larger batches of information during the update phase via fast matrix multiplication. 
For getting this speed up we need to modify the algorithm and analysis of \cite{BrandNS19}, as their algorithm has different layers that need to be handled separately in our case. 

Finally, in \Cref{sec:overview:other} we discuss how using further resparsifications we can obtain near linear size additive spanners with applications in MSSP, APSP, and diameter approximation.

\subsection{Dynamic Emulators via Low-Recourse Hitting Sets}\label{sec:overview for first emulator}
\label{sec:overview:emulator}
\paragraph{Deterministic $(1+\epsilon, 2)$-emulator and $(1+\epsilon)$-SSSP.} 
We start with a deterministic algorithm for maintaining a $(1+\epsilon,2)$-emulator. This algorithm is inspired by a randomized algorithm (working against an oblivious adversary) used by \cite{HKN2013} in the decremental setting, which in turn is based on the purely additive static construction of~\cite{DorHZ00}.
Given an unweighted graph $G=(V,E)$, we maintain an emulator $H$ with size $\tilde{O}(n^{3/2})$ as follows:
\begin{enumerate}[nosep]
    \item Let $d= \sqrt{n}$ be a degree threshold. For any node $v$ where $\deg(v) < d$, add all the edges incident to $v$ to $H$. These edges have weight $1$.
    \item Construct a \textit{hitting set} $A \subseteq V$ of size $\tilde{O}(\sqrt{n})$, such that every node with degree at least $d$, called a \textit{heavy node}, has a neighbor in $A$.
    \item For any node $u \in A$, add an edge to all nodes within distance $\lceil 2/\epsilon \rceil + 1$ to $u$. Set the weight of such an edge $(u,w)$ to $d_G(u,w)$. \label{bullet:lowhop}
    \end{enumerate}

It is easy to see that if we were interested in a randomized algorithm that only works against an oblivious adversary, we could simply construct a hitting set $A$ by uniformly sampling a fixed set of size $\tilde{O}(\sqrt{n})$~\cite{UllmanY91}. 

We could then maintain the corresponding $(\lceil 2/\epsilon \rceil + 1)$-bounded distances for all pairs in $A \times V$ after each update using the algebraic data structure by \cite{Sankowski05} which runs in $O(n^{1.529}\epsilon^{-1} \log\epsilon^{-1})$ time per update.

The distance bound of $(\lceil 2/\epsilon \rceil + 1)$ in our emulator algorithms leverages the power of algebraic distance maintenance data structures because their running times scale with the given distance bound.
However, these ideas alone are not enough for obtaining an efficient deterministic algorithm. We will have to change both the hitting set construction and the algebraic data structure.

Before explaining how to maintain both the hitting set and the corresponding distances deterministically, let us sketch the properties of this emulator and how it can be used for maintaining $(1+\epsilon)$-SSSP. It is easy to see that $H$ has size $\tilde{O}(n^{3/2})$: we add $\tilde{O}(nd)$ edges incident to low-degree nodes, and $\tilde{O}(n^{3/2})$ edges in $A \times V$. For the stretch analysis, consider any pair of nodes $s,t$, and let $\pi$ be the shortest path between $s,t$. 
We can divide $\pi$ into segments of equal length $\lceil 2/\epsilon \rceil$, and possibly one additional smaller segment. Consider one such segment $[u,v]$. If all the nodes on this segment are low-degree, then we have included all the corresponding edges in the emulator. Otherwise there is a node $w \in A$ that is adjacent to the first heavy node on this segment. We have $d_G(w,v) \leq \lceil 2/\epsilon \rceil$, and thus in the third step of the algorithm we have added a (weighted) edge $(w,v)$ in the emulator. It is easy to see that the path going through $w$ either provides a $(1+\epsilon)$ multiplicative factor, or (for the one smaller segment) an additive term of $2$. 

Given a $(1+\frac{\epsilon}{2}, 2)$ emulator, we can now maintain $(1+\epsilon)$-SSSP by (i) using algebraic techniques to maintain $O(1/\epsilon)$-bounded distances from the source $s$ to all nodes in $V$, and (ii) statically running Dijkstra's algorithm on the emulator in time $\tilde{O}(n^{3/2})$, and finally (iii) taking the minimum of the two distance values for each pair $(s,v) \in \{s\}\times V$. We observe that if $d_G(s,v) \leq O(\frac{1}{\epsilon})$, then we are maintaining a correct estimate in step~(i). Otherwise in step~(ii) the combination of the $(1+\frac{\epsilon}{2})$ multiplicative factor and the additive term, leads to an overall $(1+\epsilon)$-approximate estimate.

\paragraph{Deterministic low-recourse hitting set.} As discussed, we can easily obtain a fixed hitting set of size $ \tilde O (n/d) $ using randomization, but we are interested in a deterministic algorithm. One natural approach for constructing the hitting set $A$ deterministically is as follows: For each node $v$ with degree at least $d$, consider a set of exactly $d$ neighbors of $v$. After \textit{each update} we can \textit{statically} and \textit{deterministically} compute an $O(\log n)$-approximation to this instance of the hitting set problem. We use a simple greedy algorithm that proceeds by sequentially adding nodes to $A$ that hit the maximum number of uncovered heavy nodes.
 
 This can be done in $\tilde{O}(nd)$ time and gives us a hitting set of size $ \tilde O (n/d) $ as well. This running time is within our desired update-time bound, but we also need to maintain $\Theta(1/\epsilon)$-bounded distances from elements in this hitting set. 
 As we outline in Appendix \ref{sec:batchquerycomparison},
 by using the naive approach of recomputing a hitting set in each update and employing off-the-shelf algebraic data structures (e.g.~\cite{Sankowski05,BrandNS19}) for maintaining bounded distances in $A \times V$, we would get an update time of $O(n^{1.596})$ for current $\omega$. However, there is a 
 conditional lower bound of $O(n^{1.529})$ for this problem \cite{BrandNS19}, and our goal is to design an algorithm that matches this bound. 

To get a better running time, we change both our construction and the algebraic data structure (see \Cref{sec:overview:algebraic}) to use a \textit{low-recourse hitting set} instead, which ensures that in each update only a constant number of nodes are added to the set. More formally in Section \ref{sec:hitting_set} we will prove the following lemma:
\begin{restatable}{lemma}{hittingset} \label{lem:hitting-set}
Given a graph $G=(V,E)$ undergoing edge insertions and edge deletions and a degree threshold $d$, call a node $v$ heavy if it has degree at least $d$. We can deterministically maintain a hitting set $A_d$ of size $O(\frac{n \cdot \log n}{d})$ with worst-case $O(1)$ recourse and worst-case $O(d^2 + d \log n)$ time per update (after $ O (n d) $ preprocessing time) such that all heavy nodes have a neighbor in $A_d$.
\end{restatable}

At a high-level our dynamic low recourse hitting set proceeds as follows: we start by using the static greedy hitting set algorithm.
We then note that each update (insertion or deletion) can make at most $2$ heavy nodes uncovered. We can keep on adding arbitrary neighbors of such nodes to our hitting set $A$ until the size of the hitting set exceeds its initial $ O (\tfrac{n}{d} \log n) $ bound by a constant factor, and then reset the construction. This leads to an amortized constant recourse bound, and we can then use a standard technique to turn this into a worst case constant recourse bound (see Section \ref{sec:hitting_set} for details).

Note that this hitting set problem can be seen as a set cover instance of size $O(nd)$, where each set consists of exactly $d$ neighbors of a heavy node.
Dynamic set cover approximation has received significant attention in recent years (e.g.~\cite{abboud2019, BHN2019, gupta2017, BHNW2021}). The most relevant result to our setting is a fully-dynamic $O(\log n)$-approximate set cover algorithm by \cite{gupta2017}). However we cannot use their result directly, as they state that their polynomial time algorithm only leads to constant \textit{amortized} recourse, and their update-time guarantees are also only amortized\footnote{Of course the goal in \cite{gupta2017} is a generic set cover approximation algorithm, which is why they are not comparable to our specialized algorithm. Also, the other set cover algorithms cited lead to approximation ratio dependent on an instance parameter $f$, which can be as large as $n$ in our case.}. Here we use a simple approach that utilizes the properties of our hitting set instance, which is enough to get \textit{worst-case} recourse bounds.

\paragraph{Deterministic $(1+\epsilon,4)$-emulator for $(1+\epsilon)$-$st$ distances.}
Next, we outline how we can improve the $O(n^{1.529})$ update time to $O(n^{1.407})$ in case of $st$-distances.
For this purpose, we maintain a $(1+\epsilon, 4)$-emulator with size $\tilde{O}(n^{4/3})$, which again is inspired by the purely additive construction of~\cite{DorHZ00} in the static setting, by making the following modifications to the algorithm described in Section~\ref{sec:overview for first emulator} above: We set the degree threshold to $d=n^{1/3}$. More importantly, rather than adding edges corresponding to bounded distances in $A_d \times V$, we only add pairwise edges between nodes (with bounded distance) in $A_d \times A_d$.
This has two advantages: First, we can run Dijkstra on a sparser graph. Second, the algebraic steps can be performed much faster when we only need to maintain pairwise distances between two sets of sublinear size (here $|A_d|=\tilde{O}(n^{2/3})$, rather than from a set of size $\tilde{O}(\sqrt{n})$ to \textit{all} nodes in $V$. 

It is easy to see that this emulator has size $\tilde{O}(n^{4/3})$. There are $\tilde{O}(nd)$ edges corresponding to low-degree nodes, and $\tilde{O}(n^{4/3})$ corresponding to edges in $A_d \times A_d$. The stretch argument follows a similar structure to the one for the $(1+\epsilon,2)$-emulator. Again, for each pair of nodes $s,t$, we divide the shortest path to segments of equal length $\Theta(1/\epsilon)$. The main difference is that here we should consider the first and last heavy nodes on each segment, which we denote by $x$ and $y$. Then there must be nodes $w_1, w_2 \in A_d$ that are adjacent to $x$ and $y$ respectively. We have $d_G(w_1,w_2) \leq \Theta(1/\epsilon)$ and thus we have added an edge $(w_1,w_2)$ in the emulator. The path using this edge will lead to either a $(1+\epsilon)$-multiplicative stretch for this segment, or an additive term of $4$ for the (at most) one smaller segment.

Note that this algorithm does not lead to better bounds for single-source distances since querying $ \Theta(1/\epsilon)$-bounded single-source distances still takes $O(n^{1.529})$ time using known algebraic techniques. However, if we are interested in the $\Theta(1/\epsilon)$-bounded distance between a fixed pair of nodes $s$ and $t$, our algebraic approach, as outlined in \Cref{sec:overview:algebraic}, leads to better bounds. In this case, we get an improved bound of $O(n^{1.407})$.

\subsection{Dynamic Pairwise Bounded Distances via Matrix Inverse}\label{sec:overview:algebraic}
As outlined before, 
we must efficiently maintain bounded pairwise distances 
for some sets $S\times T \subseteq V\times V$, 
where the sets $S$ and $T$ are dynamically changing. 
We additionally use the fact that even though these sets change, 
they do not change substantially with each update because of our low-recourse hitting sets.
In this section, we outline the following: 
(i) a reduction from maintaining $S\times T$-distances to maintaining a submatrix%
\footnote{Throughout, we use $\mN_{S,T}$ for sets $S,T \subseteq [n]$ 
and $n\times n$ matrix $\mN$ to denote the submatrix consisting of rows with index in $S$ 
and columns with index in $T$.} 
$(\mA^{-1})_{S,T}$ for some dynamic matrix $\mA$, 
and (ii) a dynamic algorithm maintaining this submatrix of the inverse efficiently.
This dynamic matrix inverse algorithm, together with the reduction, then imply the following dynamic algorithm  (\Cref{thm:overview:low_hop}, proven in \Cref{sec:algebraic}) for maintaining bounded distances.
\begin{restatable}{theorem}{lowHopThm}
    \label{lem:overview:pairwise}
    \label{thm:overview:low_hop}
    For all $0\le\nu\le\mu\le1$ there exists a deterministic dynamic algorithm that, after preprocessing a given unweighted directed graph $G$ and sets $S,T \subseteq V$, supports edge-updates to $G$ and set-updates to $S$ and $T$ (i.e.~adding or removing a node to $S$ or $T$) as long as $|S|,|T|\le n^\mu$ throughout all updates.
    After each edge- or set-update the algorithm returns the $h$-bounded pairwise distances of $S\times T$ in $G$.

    The preprocessing time is $O(n^\omega h^2 \log h)$,
    and the worst-case update time is
    $$O((n^{\omega(1,1,\mu)-\mu} + n^{\omega(1,\mu,\nu)-\nu} + n^{\mu+\nu} + |S\times T|) h^2 \log h).$$
    For current bounds on rectangular matrix multiplication $\omega(\cdot,\cdot,\cdot)$ \cite{GallU18}, 
    this is $O((n^{1.407} + |S\times T|) h^2 \log h)$ for $|S|,|T| \le n^{0.85}$, 
    or $O((n^{1.529} + |S\times T|) h^2 \log h)$ for any (possibly larger) $S,T$.
\end{restatable}

For our approximate $st$-distance algorithm, we will set $|S|=|T|=\tilde{O}(n^{2/3})$ and $h = O(1/\epsilon)$,
resulting in $O(n^{1.407} \epsilon^{-2} \log \epsilon^{-1})$ update time.
For our approximate SSSP algorithm we will set $|S|=n$, $|T| = \tilde{O}(\sqrt{n})$ and $h = O(1/\epsilon)$,
resulting in $O(n^{1.529} \epsilon^{-2} \log \epsilon^{-1})$ update time.

\paragraph{Reducing distances to matrix inverse.}
All previous fully dynamic algebraic algorithms that maintain distances work by reducing the task to the so called ``dynamic matrix inverse'' problem \cite{Sankowski05,BrandNS19,BrandN19,BrandS19,GuR21,BHGWW2021}.
This reduction is due to Sankowski \cite{Sankowski05} who originally used the adjoint instead of the matrix inverse.
In previous work on fully dynamic algebraic algorithms, this reduction was always randomized. 
Here we recap the reduction when using matrix inverse instead of adjoint, and argue why the reduction can be derandomized for our use-case of maintaining bounded distances. 
Readers already familiar with this reduction might want to skip ahead to the paragraph labeled ``\emph{Submatrix maintenance}''.

For the reduction, we are given an adjacency matrix $\mA$. 
Note that $\mA^k_{s,t}$ (where $\mA^k$ is the $k$-th power of $ \mA $) is the number of (not necessarily simple) paths from $s$ to $t$ of length $k$. 
Specifically, the smallest $k$ with $\mA^k_{s,t} \neq 0$ is the distance from $s$ to $t$.
We can maintain these powers of $\mA$ via dynamic matrix inverse as follows:

Let $X$ be some symbol and let $(\mI-X\mA)$ be the matrix with $1$ on the diagonal and $(\mI-X\mA)_{u,v} = -X$ for all edges $(u,v) \in E$.
When performing all arithmetic operations\footnote{%
For our proofs, this is formalized as the entries of the matrix being from $\F[X]/\langle X^h \rangle$ for some field $\F$, i.e.~polynomials over $\F$ where we truncate all monomials of degree $\ge h$.} modulo $X^h$, we have $(\mI-X\mA)^{-1} = \sum_{k=0}^{h-1} X^k \mA^k$. To see this, observe
$$(\mI-X\mA)\cdot\sum_{k=0}^{h-1} X^k \mA^k = \sum_{k=0}^{h-1} X^k \mA^k - \sum_{k=1}^{h} X^k \mA^k = \mI$$
where the last identity holds by $X^{h} \mA^k = 0$ because of the entry-wise mod $X^h$.
Thus, a dynamic algorithm that maintains the inverse of matrix $(\mI-X\mA)$ 
is able to maintain distances of length $< h$ in dynamic graphs.
The task of maintaining pairwise distances for $S\times T$ thus reduces to the task of maintaining the submatrix $(\mM^{-1})_{S,T}$ for some dynamic matrix $\mM$.

Note that the number of $uv$-paths of length $k$, given by $\mA^k_{u,v}$, might be as large as $O(n^k)$.
Representing this number needs $O(k)$ words in Word-RAM model and each arithmetic operation needs $O(k)$ time \cite{Knuth97}.
In general, a graph might have paths of length $O(n)$, thus randomization was used in previous work \cite{Sankowski05,BrandNS19,BrandN19,BrandS19,GuR21,BHGWW2021} 
to bound the bit-length and arithmetic complexity of the numbers involved (e.g.~by maintaining the number of paths modulo some small random prime $p = \poly(n)$, or by using Schwartz-Zippel lemma).\footnote{We focus on fully dynamic algorithms here. We note that in the \emph{incremental} setting (i.e.~only edge insertions), such randomization is not required. See e.g.~\cite{KarczmarzMS22}.}

However, in our use-case, we only need distances up to $O(1/\epsilon)$ thanks to properties of our emulators, thus the randomization is not required.
Each arithmetic operation will only need $O(1/\epsilon)$ time as we only consider numbers represented by $O(1/\epsilon)$ words.

\paragraph{Submatrix maintenance.}

As explained in the previous paragraph, our dynamic distance algorithms reduce to a dynamic matrix inverse algorithm that maintains a submatrix $\mM^{-1}_{S,T}$ for some dynamic matrix $\mM$.
Any existing dynamic matrix algorithm can maintain such a submatrix by just querying all $|S\times T|$ entries after each change to $\mM$, but this would not be fast enough for our purposes.
We instead propose a new dynamic matrix inverse algorithm that can maintain such a submatrix efficiently, if the sets $S$ and $T$ are slowly changing.

The construction of this dynamic algorithm relies on reducing maintaining $\mM^{-1}_{S,T}$ to maintaining partial rows of the form $(\mM^{-1})_{k,T}$ for any $k\in[n]$, formalized in \Cref{lem:reduction_multi} (proven in \Cref{sec:algebraic}).

\begin{restatable*}{lemma}{reductionMulti}\label{lem:reduction_multi}
Assume we are given a dynamic algorithm 
that initializes on a dynamic set $T\subset[n]$ and a dynamic $n\times n$ matrix $\mM$ that is promised to stay non-singular.
Assume the algorithms supports both changing any entry of $\mM$ and adding/removing any index to/from $T$ in $O(u(|T|,n))$ operations,
and supports queries for any $i\in[n]$ that return $\mM^{-1}_{i,T}$ in $O(q(|T|,n))$ operations.

Then the dynamic algorithm can also maintain $\mM^{-1}_{S,T}$ explicitly for dynamic matrix $\mM$ and dynamic sets $S,T\subset[n]$
while the update time increases to $O(u(k,n)+q(k,n)+|S\times T|)$ for $k=\max(|S|,|T|)$.
The preprocessing time increases by an additive $O(n^\omega)$ operations.
\end{restatable*}

Thus it suffices to design a dynamic matrix inverse algorithm that supports efficient queries to partial rows $\mM^{-1}_{i,T}$.

Our proposed algorithm is a modified version of the dynamic matrix inverse algorithm by \cite{BrandNS19}. 
Their data structure has the fastest known update complexity among all dynamic matrix inverse algorithms, but comes at the cost of slower queries than some data structures from \cite{Sankowski04}.

We are able to accelerate the queries of \cite{BrandNS19} by exploiting the fact that set $T$ is slowly changing, thus we know ahead of time which entries of the inverse might be queried in the future. 
By preprocessing these entries, we can speed up queries to $\mM^{-1}_{i,T}$ for any $i\in[n]$ and a dynamic set $T\subset[n]$.
In addition to these faster queries, we also simplify the proof and the structure of the dynamic algorithm from \cite{BrandNS19}.

We next explain how to achieve such a speed up. We start with a quick recap of how the data structure of \cite{BrandNS19} represents the dynamic matrix inverse and then explain how we modify the algorithm.
Let $\mM'$ be the dynamic matrix $\mM$ during initialization, then we maintain $\mM$ in the following implicit form:
\begin{align}
\mM = \mM' + \mU' \mV'^\top + \mU \mV^\top \label{eq:overview:invariant}
\end{align}
where for some $0\le\nu\le\mu\le1$, the matrices $\mU',\mV'$ have at most $n^\mu$ columns and $\mU,\mV$ have at most $n^\nu$ columns, 
all of which have at most one non-zero entry per column.
Initially, $\mU,\mU',\mV,\mV'$ are all empty matrices (i.e. with 0 columns) as $\mM = \mM'$.
Then, with each update to $\mM$, we update $\mU$ and $\mV$ as follows:
The entry update to $\mM_{i,j}$ can be represented as adding some $v\cdot e_i e_j^\top$ to $\mM$ for some scalar $v$. 
We can thus maintain \eqref{eq:overview:invariant} by setting 
$\mU \leftarrow [\mU|v\cdot e_i]$ and $\mV \leftarrow [\mV| e_j]$ 
(i.e.~appending a new column to $\mU$ and $\mV$).
After $n^\nu$ updates, the matrices $\mU$ and $\mV$ have $n^\nu$ columns 
and we append these columns to $\mU',\mV'$ by setting $\mU' \leftarrow [\mU'|\mU]$, $\mV' \leftarrow [\mV'|\mV]$,
then we reset $\mU,\mV$ to be empty matrices (i.e. with 0 columns).
Thus $\mM$ is still maintained in form \eqref{eq:overview:invariant} 
and we can assume $\mU$, $\mV$ always have at most $n^\nu$ columns.
After $n^\mu$ updates, the algorithm is reset by letting $\mM' \leftarrow \mM$ and all $\mU,\mU',\mV,\mV'$ are reset to be empty matrices. Thus we can also assume $\mU',\mV'$ have at most $n^\mu$ columns.

The task is now to maintain $\mM^{-1}$ in some implicit form that allows for fast queries to $\mM^{-1}_{i,T}$ for any $i\in[n]$ and a dynamic set $T\subset[n]$.
For this consider the following Sherman-Morrison-Woodbury identity. 

\begin{lemma}[{\cite{ShermanM50,Woodbury50}}]
\label{lem:woodbury}
For any non-singular $\mM$ and $\mM+\mU\mV^\top$
we have
$$
(\mM + \mU \mV^\top)^{-1}
=
\mM^{-1} - \mM^{-1} \mU (\mI + \mV^\top \mM^{-1} \mU)^{-1} \mV^\top \mM^{-1}.
$$
\end{lemma}

By applying this identity twice (once for $\mM:=\mM''+\mU\mV^\top$ and once for $\mM'':=\mM'+\mU'\mV'^\top$) we can write:
\begin{align}
\mM^{-1} = \underbrace{\mM'^{-1} + \mA \mB}_{=\mM''^{-1}} +&~ \mM''^{-1}\mU \mC \mV^\top \mM''^{-1} \label{eq:maintenance}\\
\mA := \mM'^{-1} \mU' (\mI + \mV'^\top \mM'^{-1} \mU')^{-1}, \quad \mB :=&~\mV'^\top \mM'^{-1}, \quad \mC := (\mI + \mV^\top \mM''^{-1} \mU)^{-1} \notag
\end{align}
where matrices $\mA,\mB,\mC$ are maintained by the data structure.
In \cite{BrandNS19}, it was shown that this representation (that is, matrices $\mM'^{-1},\mA,\mB,\mU,\mV,\mC$) can be maintained in $O(n^{\omega(1,1,\mu)-\mu} + n^{\omega(1,\mu,\nu)-\nu}+n^{\mu+\nu})$ time per update.\footnote{%
Technically, \cite{BrandNS19} uses \Cref{lem:woodbury} in the form $(\mM+\mU\mV^\top)^{-1} = \mM^{-1}\mT$ for $\mT = (\mU(\mI+\mV^\top\mM^{-1}\mU)^{-1}\mV^\top\mM^{-1})$. 
Then sum \eqref{eq:maintenance} is written as a matrix product of two such $\mT$, one for $\mU\mV^\top$ and one for $\mU'\mV'^\top$. 
In \Cref{sec:algebraic} we reprove the algorithm in sum-form \eqref{eq:maintenance} which simplifies both the analysis of the algorithm and the analysis of our modifications to accelerate the queries.}

\begin{figure}
\center
\includegraphics[trim={0 200 110 50},clip,scale=0.5]{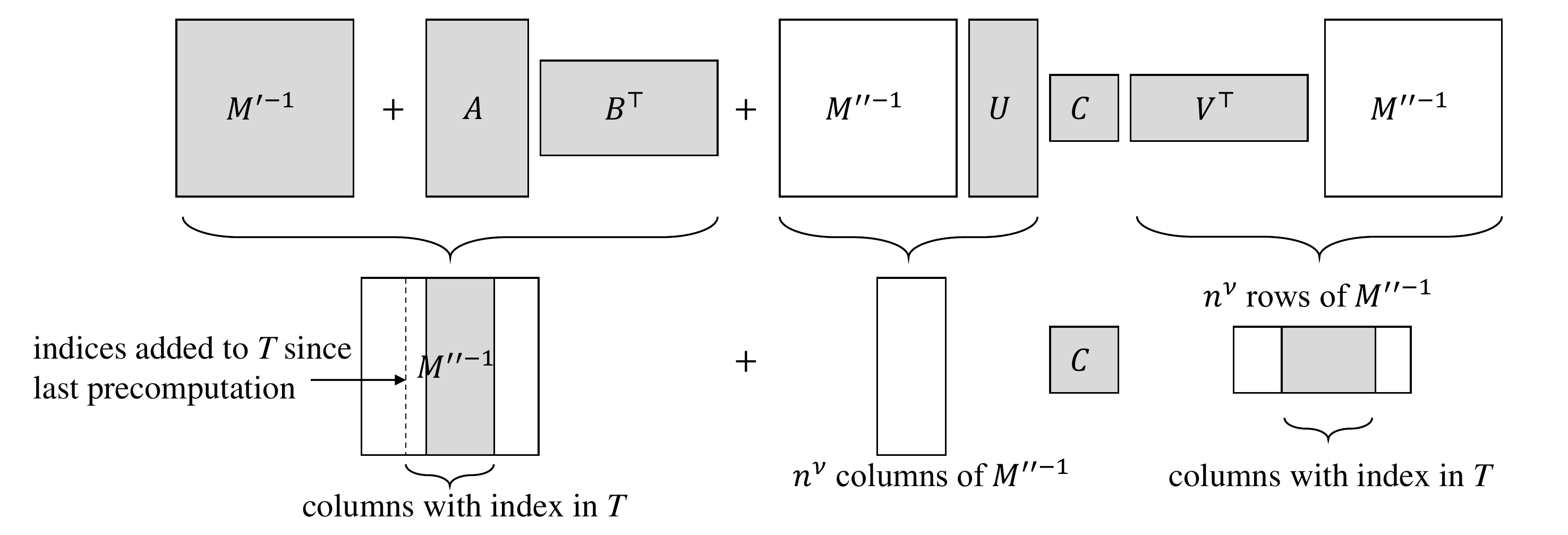}
\caption{\label{fig:maintenance}%
Implicit representation of $\mM^{-1}$ \eqref{eq:maintenance}. Gray blocks represent (sub-)matrices that are explicitly maintained. White (sub-)matrices are only implicitly accessible (i.e.~entries can be queried). Matrices $\mA,\mB$ have $n^\mu$ columns while $\mU,\mV,\mC$ have $n^\nu$ columns for $0\le\nu\le\mu\le1$. Matrices $\mU,\mV$ have only one non-zero entry per column.
}
\end{figure}

Here the representation of $\mM^{-1}$ and $\mM''^{-1}$ is only implicit via \eqref{eq:maintenance}, while matrices $\mM'^{-1}$, $\mA$, $\mB$, $\mU$, $\mV$, $\mC$ are known explicitly (i.e.~direct read access in memory). The first row of \Cref{fig:maintenance} shows \eqref{eq:maintenance} where each box represents one of the matrices and gray matrices are computed explicitly.
We modify the algorithm by computing some submatrices of the implicit $\mM''^{-1}$ and $\mV^\top\mM''^{-1}$ explicitly (see gray areas in the second row of \Cref{fig:maintenance}). 
Every time matrices $\mA$ and $\mB$ change (i.e.~every $n^\nu$ iterations) we precompute $\mM''^{-1}_{[n],T} = \mM'^{-1}_{[n],T} + (\mA \mB^\top)_{[n],T}$ for current set $T$.
It is possible to show that this precomputation can be performed in $O(n^{\omega(1,\mu,\nu)})$ operations if $|T| \le n^\mu$.
Since $T$ is slowly changing, whenever we attempt to query $\mM''^{-1}_{i,T}$ at a later point, there are at most $O(n^\nu)$ entries that have not been precomputed yet. Each of these missing entries can be computed in $O(n^\mu)$ time because $\mM''^{-1}_{i,j} = \mM'^{-1}_{i,j} + (e_i^\top \mA) (\mB e_j)$ where $\mA$ and $\mB$ have at most $n^\mu$ columns.
Thus any row $\mM''^{-1}_{i,T}$ can be obtained in $O(n^{\nu+\mu})$ operations.

With every update to $\mM$, we also maintain the columns of $\mV^\top\mM''^{-1}$ with index in $T$. 
Note that by $\mV$ having at most $n^\nu$ columns, each with only one non-zero entry, 
$\mV^\top\mM''^{-1}$ are just $\le n^\nu$ rows of $\mM''^{-1}$.
Further, with each update to $\mM$, $\mV$ grows by one column, 
so one more row of $(\mM''^{-1})_{i,[n]}$ is added to $\mV^\top\mM''^{-1}$ for some $i\in[n]$. 
So we can maintain the desired submatrix of $\mV^\top\mM''^{-1}$ by querying the entries $\mM''^{-1}_{i,T}$ in $O(n^{\mu+\nu})$ operations.
If an index is added to $T$, we need to compute one new column of $\mV^\top\mM''^{-1}$, 
which means we just need to query $\le n^\nu$ entries of $\mM''^{-1}$.
This can also be done in $O(n^{\nu+\mu})$ operations.

With these explicit submatrices maintained (see \Cref{fig:maintenance} for a summary), we can now query any $\mM^{-1}_{i,T}$ efficiently as follows:
Query $\mM''^{-1}_{i,T}$ in $O(n^{\nu+\mu})$ operations,
then query $(\mM''^{-1}\mU\mC\mV^\top\mM''^{-1})_{i,T}$.
For the latter, note that $e_i^\top\mM''^{-1}\mU$ are just $n^\nu$ entries of $\mM''^{-1}$ because $\mU$ has only one non-zero entry per column, and the columns with index in $T$ of $\mV^\top\mM''^{-1}$ are maintained explicitly.
Thus, this also takes just $O(n^{\nu+\mu})$ operations by $|T|\le n^\mu$.

In summary, our modification has amortized complexity (which can be made worst-case via standard techniques, see e.g.~\cite[Theorem B.1]{BrandNS19})
\begin{align*}
O(\underbrace{n^{\omega(1,\mu,\nu)-\nu}}_{
\begin{array}{l}
\scriptsize\text{Explicitly maintain}\\
\scriptsize\text{submatrix of $\mM''^{-1}$.}
\end{array}} + \underbrace{n^{\mu+\nu}}_{
\begin{array}{l}
\scriptsize\text{Explicitly maintain}\\
\scriptsize\text{submatrix of $\mV''^\top\mM''^{-1}$}\\
\scriptsize\text{+query any $\mM''^{-1}_{i,T}$}
\end{array}})
\end{align*}
This is subsumed by the complexity of \cite{BrandNS19}
for maintaining the matrices $\mM'^{-1},\mA,\mB,\mC$
in \eqref{eq:maintenance}.
So our modification of their algorithm does not increase the update complexity despite precomputing submatrices of $\mM''^{-1}$ and $\mV'^\top\mM''^{-1}$.

\subsection{Sparse Emulators, MSSP, APSP, and Further Applications} 
\label{sec:overview:other}
\paragraph{Sparser emulators with applications in APSP and MSSP.} Finally, we give another algorithm that lets us maintain much sparser emulators, which further leads to improvements when we need to maintain approximate distances from many sources (e.g. MSSP and APSP).

We start by maintaining near-linear size emulators as follows: first \textit{maintain} a $(1+\epsilon, 4)$-emulator $H_1$ of $G$. Then statically construct a much sparser $(1+\epsilon, n^{o(1)})$-emulator $H_2$ of size $\tilde{O}(n^{1+o(1)})$. The key idea here is to use $H_1$ in order to construct $H_2$ more efficiently.
We use a \textit{static} deterministic emulator algorithm (based on \cite{RTZ2005, TZ2006emulators}) that can construct such an emulator in time $O(|E(H_1)|n^{o(1)})$.
This leads to a fully-dynamic algorithm for maintaining $(1+\epsilon, n^{o(1)})$-emulators deterministically in $\tilde{O}(n^{1.407})$ worst-case update time.

Now we can use this to maintain multi-source distances from many (up to $O(n^{0.52})$) sources with an update time almost the same as the time required for single-source distances. 
For this purpose, given a set of sources $S$, we use the above approach to maintain an emulator of size $\tilde{O}(n^{1+o(1)})$.
Then, similar to before, in each update we find distance estimates for pairs in $S \times V$ by computing the minimum of the following estimates: i) $n^{o(1)}$-bounded distances from all sources maintained by an algebraic data structure, ii) distances from all sources on emulator $H_2$, computed in time $O(|S| \cdot n^{1+o(1)})$.

This lets us maintain $(1+\epsilon)$-MSSP from up to $O(n^{0.52})$ sources in almost (up to an $n^{o(1)}$ factor) the same running time as maintaining distances from a single-source by computing multi-source distances statically on this very sparse emulator and querying small distances from the algebraic data structure. 
This approach naturally extends to maintaining all-pairs distances deterministically and yields a worst-case update time of $ n^{2 + o(1)} $ by setting $S=V$.

Having described our approach for maintaining the more general emulator, let us briefly explain the differences to the dynamic spanner algorithm of~\cite{BHGWW2021}: We do not aim at \emph{directly} maintaining an almost linear-size spanner. Instead, we use a two-level scheme in which we first compute a $(1+\epsilon, 4)$-emulator of ``medium'' sparsity (outlined in \Cref{sec:overview:emulator}) and then resparsify this first-level emulator with a static algorithm. 
Hence, we get improved bounds for the second-level (near-linear size) emulators, since we can maintain the ``first-level'' emulators more efficiently than the algorithms in~\cite{BHGWW2021} due to the properties described in \Cref{sec:overview:emulator}. Moreover, our deterministic dynamic hitting set and our algebraic data structure supporting its changes let us maintain these emulators deterministically, whereas the spanners of~\cite{BHGWW2021} are randomized.

\paragraph{Diameter Approximation.} 
Our sparse emulators can also be used to maintain a (nearly) $(3/2+\epsilon)$-approximation of the diameter. Our algorithm is an adaptation of the dynamic algorithm by \cite{BrandN19}, which is in turn based on an algorithm by \cite{RV13}. At a high-level, we need to query (approximate) multi-source distances from three sets of size at most $O (\sqrt{n})$. We show that our emulators can be used to maintain such approximate distances much more efficiently than the data structures of \cite{BrandN19}. 

\paragraph{$(1+\epsilon)$-APSP Distance Oracles.} Finally, we maintain a data structure with worst-case subquadratic update time that supports sublinear all-pairs $(1+\epsilon)$-approximate distance queries. Our algorithm is based on ideas of \cite{RZ12,BrandN19} that utilize well-known path hitting techniques (e.g.~\cite{UllmanY91}). In order to get improved bounds we again use our sparse $(1+\epsilon, \beta)$-emulators. We need to handle some technicalities both in the algorithm and its analysis introduced by the additive factor~$\beta$, combined with $h$-bounded distances maintained in the algorithm of \cite{BrandN19} for an appropriately chosen parameter $h$.

\section{Approximate Distances via Emulators} \label{sec:unweighted}
\label{sec:emulator}
In this section we focus on maintaining emulators with various tradeoffs and describing how they can be combined with the algebraic data structure of Lemma~\ref{lem:overview:pairwise} for obtaining dynamic $st$ and single-source distance approximations.
While our main focus is on $st$-distances, as a warm-up we start with our SSSP result. 

We first assume that we have a low-recourse dynamic hitting set which we use in maintaining $(1+\epsilon, 2)$-emulators (with application in $(1+\epsilon)$-SSSP) and $(1+\epsilon, 4)$-emulators (with applications in $(1+\epsilon)$-$st$). 
We will then move on to give a deterministic algorithm that maintains low-recourse hitting sets.

\subsection{Deterministic $(1+\epsilon, 2)$-Emulators and $(1+\epsilon)$-SSSP}\label{sec:dense_emulator}
\label{sec:emulator:sssp}

In this section we describe how to maintain $(1+\epsilon)$-SSSP with a worst-case update time matching the conditional lower bound of~\cite{BrandNS19}. We start by describing how to maintain a $(1+\epsilon, 2)$-emulator, assuming that we have a low-recourse hitting set, and can compute bounded-hop distances from elements in this set. The algorithm is summarized in Algorithm \ref{alg:2_emulator}. Assume that we are given two functions: 
\begin{itemize}
    \item $\textsc{UpdateHittingSet}(G, d)$, which returns a dynamically maintained hitting set for neighborhoods of heavy nodes (i.e., with degree at least $ d $) satisfying Lemma \ref{lem:hitting-set}. We provide an efficient algorithm for this function in Section \ref{sec:hitting_set}.
    \item $\textsc{QueryDistances}(G, S, T, h)$, which can query $h$-bounded distances between pairs in $S \times T$ as specified in Lemma \ref{lem:overview:pairwise}. In Section \ref{sec:algebraic}, we formally explain how these distances can be maintained and then queried for our low-recourse hitting sets. 

\end{itemize}

\begin{algorithm2e}[h]
\Input{Unweighted Graph $G = (V, E)$\;}
 $A_d:= \textsc{UpdateHittingSet}(G, d)$ with $d= \sqrt{n\log n}$\;  
For all nodes $\{ v: \deg(v) \leq d \}$, add all the edges incident to $v$ to $H$ with weight $1$\;
$\textsc{QueryDistances}(G, A_d, V, \lceil \frac{2}{\epsilon} \rceil+1)$\;
Add edges $\{ (u,w) : u \in A_d, w\in V, d_G(u,v) \leq \lceil \frac{2}{\epsilon} \rceil+1\}$ to $H$, and set the weight of each edge $(u,w)$ to $d_G(u,w)$ \;
\Return{$H$}\;
\caption{Update Algorithm for a $(1+\epsilon, 2)$-Emulator}\label{alg:2_emulator}
\end{algorithm2e}

Observe that even though we start with an unweighted graph, we need to add weighted edges to the emulator (with weight corresponding to the distance between the endpoints). We note that a similar, but randomized version of this emulator construction (working only against an oblivious adversary) was used in \cite{HKN2013} for maintaining approximate shortest paths decrementally. For completeness we provide  a full analysis of the properties of this emulator here. Assuming that we can maintain hitting set $A_d$ for $d= \sqrt{n \log n}$ satisfying Lemma \ref{lem:hitting-set} and $O(1/\epsilon)$-bounded distances in $A_d \times V$, Algorithm \ref{alg:2_emulator} can be used to show the following theorem:

\begin{theorem}\label{thm:dense_emulator}
Given an unweighted graph $G=(V,E)$, $0 < \epsilon <1$, we can deterministically maintain a $(1+\epsilon, 2)$-emulator with size $O(n^{3/2} \sqrt{\log n})$. 
The worst-case update time is 
$O((n^{\omega(1,1,\nu)-\nu}+n^{1+\nu}) \epsilon^{-2}\log \epsilon^{-1})$
for any $0\le\nu\le1$
and preprocessing time is $O(n^{\omega} \epsilon^{-2}\log \epsilon^{-1})$.
\end{theorem} 
For current bounds on $\omega$ and best choice of $\nu\approx0.529$, this is $O( n^{1.529} \epsilon^{-2}\log \epsilon^{-1})$ update time.
\begin{proof}
The size analysis is straightforward. We set $d= \sqrt{n \log n}$, and add $O(nd)$ edges for sparse nodes, and by \Cref{lem:hitting-set}, we have $O(\sqrt{n/ \log n})$ nodes in the hitting set and thus we add an overall $O(n^{3/2} \sqrt{\log n})$ edges for all of them.

We next move on to the stretch analysis. Consider any pair of nodes $s,t \in V$ and let $\pi$ be the shortest path between $s$ and $t$ in $ G $.
We divide $\pi$ into segments of length exactly $\lceil 2/\epsilon \rceil$ and possibly one shorter segment that we handle separately (which could be the only segment if $d_G(s,t) \leq \lceil 2/\epsilon \rceil $). We show that the emulator $ H $ contains for each segment of length $\lceil 2/\epsilon \rceil$ a path of multiplicative stretch $(1+\epsilon)$ and for the shorter segment a path of additive stretch $ 2 $.

Consider the $i$-th segment that we denote by $[u_i, u_{i+1}]$, and let the corresponding shortest path between $u_i$ and $u_{i+1}$ be $\pi'$. If all the nodes on $\pi'$ have degree less than $d$, then all the edges of the segment are in $H$.
Otherwise, let $v$ be the first heavy node on $\pi'$.
By Lemma \ref{lem:hitting-set} we know that there is a node $w \in A_d$ adjacent to $v$.
First assume that $d_G(u_i, u_{i+1}) = \lceil 2/\epsilon \rceil$. Since $d_G(u_i,u_{i+1}) \leq \lceil 2/\epsilon \rceil$ and the neighbor $v$ of $ w $ is on the shortest path between $u_i$ and $u_{i+1}$, we have $d_G(w,u_{i+1}) \leq \lceil 2/\epsilon \rceil + 1$.
Therefore, we have added an emulator edge between $w$ and $u_{i+1}$ to $ H $. Consider the path in $H$ going through $u_i \rightarrow v \rightarrow w \rightarrow u_{i+1}$. For the length of this path we have
\begin{align}
   d_H(u_i,u_{i+1}) &\leq d_G(u_i,v)+1+w_H(w,u_{i+1}) = d_G(u_i,v)+1+d_G(w,u_{i+1}) \\
   &\leq d_G(u_i,v)+1+d_G(v,u_{i+1})+1 \\
   &\leq d_G(u_i,u_{i+1})+ 2 \label{eq:add2} \\
   &\leq d_G(u_i,u_{i+1}) + \epsilon d_G(u_i,u_{i+1})\\ &\leq (1+\epsilon) d_G(u_i,u_{i+1}) \, .
\end{align}

Now assume that $d_G(u_i,u_{i+1}) < \lceil 2/\epsilon \rceil$.
Then the same analysis as in the previous case works up to inequality \eqref{eq:add2} and we thus have we have $d_H(u_i,u_{i+1}) \leq d_G(u_i,u_{i+1})+2$.
Hence, for any pair of nodes $s$ and $t$, the overall multiplicative stretch in $ H $ with respect to $ G $ is $(1+\epsilon)$ together with an additive stretch of~$2$.
 
The running time (update time and preprocessing) follows from Lemma \ref{lem:hitting-set} for maintaining a low-recourse hitting set $A_d$, 
and Theorem~\ref{lem:overview:pairwise}
for maintaining $O(1/\epsilon)$-bounded distances in $A_d \times V$, by setting $S= A_d$ and $T=V$. Note that we also need to maintain $d$ edges incident to heavy nodes that overlap with any node added to $A_d$, but this only takes $O(d)$ time per update.
\end{proof}

\paragraph{Using an emulator for maintaining $(1+\epsilon)$-SSSP.}
Given an unweighted graph $G=(V,E)$, we first maintain a $(1+\frac{\epsilon}{2}, 2)$-emulator $H$ for $G$. Given a single-source $s$ and $H$, we can now maintain the distances by:
\begin{enumerate}
    \item [(i)] Using the algebraic data structure of Lemma~\ref{lem:overview:pairwise}: $\lceil 4/\epsilon \rceil$-hop bounded distances from $s$ on $G$,
    \item [(ii)] After each update, \textit{statically} computing SSSP on $H$ in $O(n^{3/2} \sqrt{\log n})$ time. 
\end{enumerate}

\sssp*
\begin{proof}
The distance estimate stored at each node is the minimum of the two estimates (i) and~(ii) described above. To see the correctness (stretch), we simply observe that for any node $v$ where $d_G(s,v) > 4/\epsilon$, we have:

\[ d_H(s,v) \leq \left(1+\frac{\epsilon}{2}\right) d_G(s,v)+ 2 \leq \left(1+ \frac{\epsilon}{2}\right) d_G(s,v)+ \frac{\epsilon}{2} \cdot d_G(s,v) \leq (1+\epsilon)d_G(s,v)\]

Hence a $(1+\epsilon)$-approximate estimate is returned due to step (ii) above. On the other hand, if $d_G(s,v) \leq 4/\epsilon$, then we directly maintained the exact distance in step~(i).

Next, we analyze the running time. The update time for maintaining the emulator, as described in Theorem \ref{thm:dense_emulator} is $O((n^{\omega(1,1,\nu)-\nu}+n^{1+\nu})\epsilon^{-2}\log \epsilon^{-1})$. We then use Theorem~\ref{lem:overview:pairwise} again for maintaining $\lceil 4/\epsilon \rceil$-hop bounded distances from the source. Additionally, we statically compute single source distance on an emulator of size $\tilde O(n^{3/2})$ in $\tilde O(n^{3/2})$ time.
\end{proof}

\subsection{Deterministic $(1+\epsilon,4)$-Emulator and $(1+\epsilon)$-$st$ Distances}
\label{sec:emulator:st}
In this section, we give another emulator-based algorithm that lets us maintain the approximate distance from a given source $s$ to a given destination $t$ with better update time than the time bound we showed for SSSP. 
We maintain a sparser emulator with a slightly larger additive stretch that supports faster computation of the $st$ distance approximation. In particular, we maintain a $(1+\epsilon, 4)$-emulator of size $\tilde{O}(n^{4/3})$.

Compared to the emulator of Section~\ref{sec:dense_emulator}, for this emulator construction we need to maintain bounded distances with our algebraic data structure for a smaller number of pairs of nodes, which increases efficiency.
This, combined with the fact that our emulators are sparser, leads to a faster algorithm for maintaining $(1+\epsilon)$-approximate $st$ distances.

\paragraph{$(1+\epsilon, 4)$-emulator.} We start by maintaining a sparse emulator with slightly larger additive stretch term. The algorithm is summarized in Algorithm \ref{alg:4_emulator}.

\begin{algorithm2e}[h]
\Input{Unweighted Graph $G = (V, E)$.}
 $A_d:= \textsc{UpdateHittingSet}(G, d)$ with $d= n^{1/3}\sqrt{\log n}$\;
For all nodes $\{ v: \deg(v) \leq d \}$, add all the edges incident on $v$ to $H$ with weight $1$\;
$\textsc{QueryDistances}(G, A_d, A_d, \lceil \frac{4}{\epsilon} \rceil+2)$\;
Add edges $\{ (u,w) : u,w \in A_d, d_G(u,v) \leq \lceil \frac{4}{\epsilon} \rceil+2\}$ to $H$, and set the weight of each edge $(u,w)$ to $d_G(u,w)$ \;
\Return{$H$}\;
\label{alg:4_emulator}
\caption{Update Algorithm for $(1+\epsilon, 4)$-Emulators}
\end{algorithm2e}
Assuming that we can maintain a low-recourse hitting set $A_d$ and $O(1/\epsilon)$-bounded distance between pairs $A_d \times A_d$, 
Algorithm \ref{alg:4_emulator} leads to an emulator with the following guarantees:
\begin{theorem}\label{thm:sparse_emulator}
Given an unweighted graph $G=(V,E)$, $0 < \epsilon <1$, we can deterministically maintain a $(1+\epsilon, 4)$-emulator of size $O(n^{4/3} \sqrt{\log n})$ 
with worst-case update time of 
$O((n^{\omega(1,1,\mu)-\mu}+n^{\omega(1,\mu,\nu)-\nu}+n^{\mu+\nu}+n^{4/3}) \epsilon^{-2}\log \epsilon^{-1})$
for any $0\le\nu\le\mu\le1$, 
and preprocessing time of 
$O({n^{\omega}}{\epsilon^{-2}\log \epsilon^{-1}})$.
\end{theorem} 
For current bounds on $\omega$ and $\mu\approx0.856$, $\nu\approx0.551$,
this is $O(n^{1.407} \epsilon^{-2}\log \epsilon^{-1})$ update time.
\begin{proof}
It is easy to see that Algorithm \ref{alg:4_emulator} leads to an emulator of size $O(n^{4/3} \sqrt{\log n})$. This follows from the fact that we add at most $O(n^{4/3} \sqrt{\log n})$ pairwise edges between nodes in $A_d \times A_d$, and we add $O(nd)$ edges for non-heavy node where $d= n^{1/3}\sqrt{\log n}$.

We next move on to the stretch analysis. Consider any pair of nodes $s,t \in V$ and let $\pi$ be the shortest path between $s$ and $t$ in $ G $.
We divide $\pi$ into segments of length exactly $\lceil 4/\epsilon \rceil$ and possibly one shorter segment that we handle separately (which could be the only segment if $d_G(s,t) \leq \lceil 4/\epsilon \rceil $). We show that the emulator $ H $ contains for each segment of length $\lceil 4/\epsilon \rceil$ a path of multiplicative stretch $(1+\epsilon)$ and for the shorter segment a path of additive stretch $ 4 $.

Consider the $i$-th segment that we denote by $[u_i, u_{i+1}]$, and let the corresponding shortest path between $u_i$ and $u_{i+1}$ be $\pi'$. If all the nodes on $\pi'$ have degree less than $d$, then all the edges of the segment are in $H$.
Otherwise, let $v$ be the first heavy node, and let $u$ be the last (furthest from $u_i$) heavy node on $\pi'$. 
By Lemma \ref{lem:hitting-set} we know that there are nodes $w_1,w_2 \in A_d$ such that $w_1$ is adjacent to $v$ and $w_2$ is adjacent to $u$. The case where $w_1=w_2$ is an easy special case, so let us assume $w_1 \neq w_2$. First assume that $d_G(u_i, u_{i+1}) = \lceil 4/\epsilon \rceil$. Since $d_G(u_i,u_{i+1}) \leq \lceil 4/\epsilon \rceil$ and the neighbors $v$ and $ u $ of $ w_1 $ and $ w_2 $, respectively, are on the shortest path between $u_i$ and $u_{i+1}$, we have $d_G(w_1,w_2) \leq \lceil 4/\epsilon \rceil+2$. Therefore, we have added an emulator edge between $w_1$ and $w_2$ to $ H $. Consider the path in $H$ going through $u_i \rightarrow w_1 \rightarrow w_2 \rightarrow u_{i+1}$. For the length of this path we have
\begin{align}
   d_H(u_i,u_{i+1}) &\leq d_G(u_i,v)+1+w_H(w_1,w_2)+d_G(w_2,u_{i+1})\\ 
   &=  d_G(u_i,v)+1+d_G(w_1,w_2)+d_G(w_2,u_{i+1}) \\
   &\leq d_G(u_i,v)+(1+d_G(v,w_2)+1)+d_G(w_2,u_{i+1}) \\
   &\leq d_G(u_i,v)+(1+d_G(v,u)+1+1)+(d_G(u,u_{i+1})+1)\\
   &\leq d_G(u_i,u_{i+1})+ 4 \label{eq:add4}\\  
   &\leq d_G(u_i,u_{i+1}) + \epsilon d_G(u_i,u_{i+1})\\
   &\leq (1+\epsilon) d_G(u_i,u_{i+1})
\end{align}

Now assume that $d_G(u_i,u_{i+1}) < \lceil 4/\epsilon \rceil$.
Then the same analysis as in the previous case works up to inequality \eqref{eq:add4} and we thus have $d_H(u_i,u_{i+1}) \leq d_G(u_i,u_{i+1})+4$.
Hence, for any pair of nodes $s$ and $t$, the overall multiplicative stretch in $ H $ with respect to $ G $ is $(1+\epsilon)$ together with an additive stretch of~$4$.

The running time analysis now follows by applying Lemma \ref{lem:overview:pairwise} to maintain $O(1/\epsilon)$-bounded distances from the low-recourse hitting set $A_d$ obtained by Lemma \ref{lem:hitting-set}. For this, we use Theorem~\ref{lem:overview:pairwise} and set $S=T=A_d$, where $|S|=|T|= \tilde{O}(n^{2/3})$.
\end{proof}
We can now use the emulator of Theorem \ref{thm:sparse_emulator} to maintain $(1+\epsilon)$-approximate $st$-distances using the same approach as in Section \ref{sec:dense_emulator}. We set the distance estimate to be the minimum obtained by maintaining $(8/\epsilon+2)$-bounded distances on $G$ (using the second algorithm in Lemma \ref{lem:overview:pairwise}), and statically running Dijkstra's on the emulator. The formal result can be stated as follows.
 
\st*
\begin{proof}
We use Theorem \ref{thm:sparse_emulator} to maintain a $(1+\epsilon, 4)$-emulator with worst-case update time of $O(n^{1.407})$.
Then we use Theorem~\ref{lem:overview:pairwise} again for maintaining $(8/\epsilon +2)$-bounded $st$ distances. Additionally, we statically compute single source distance on an emulator of size $O(n^{4/3} \sqrt{\log n})$.
The stretch analysis is the same as in Theorem \ref{thm:sssp}.
\end{proof}
\subsection{Deterministic Dynamic Low-Recourse Hitting Sets}\label{sec:hitting_set}
\label{sec:emulator:recourse}
In this section, we focus on the deterministic maintenance of the hitting sets in order to efficiently maintain the emulators described in the previous two sections.

Let us first review a very simple static deterministic construction and later discuss how to obtain a low-recourse dynamic algorithm.
Given a graph $G=(V,E)$, we create a sparse subgraph of $G$ with size $O(nd)$, which we denote by $G_d$ as follows: For any \textit{heavy} node $v \in V$ (a node with degree at least $d$), we choose an arbitrary set of $d$ neighbors of $v$, denoted by $A_v$. 

We then statically compute an $O(\log n)$-approximate hitting set $A_d$ for $G_d$ deterministically (i.e., $ A_d $ exceeds the size of a minimum hitting set on $ G_d $ by a factor of $O(\log n)$).  Consider the following simple greedy algorithm: in each step consecutively we add to $A_d$ the node that is incident to the maximum number of heavy nodes whose neighborhood does not contain a node of $A_d$ yet. It is well-known that such an algorithm leads to the following guarantees.

\begin{lemma}[Greedy Hitting Set, \cite{johnson1974}]\label{lem:static-hitting-set}
Given a graph $G=(V,E)$, and a degree threshold $d$, we can deterministically construct a hitting set $A_d$ with size $O(\frac{n\cdot \log n}{d} )$ in $O(nd)$ time, such that every heavy node has a neighbor in $A_d$.
\end{lemma}

It is easy to see that the hitting set obtained has size $O(\frac{n\log^2 n}{d} )$: by a simple sampling procedure we know that there exists a hitting set of size $O(\frac{n \log n}{d})$, and we are using an $O(\log n)$-approximate hitting set algorithm.
A tighter analysis (e.g.~see \cite{BHK2017}) will lead to a total size of $O(\frac{n \log n}{d})$.

Next we move on to a dynamic maintenance of hitting sets. Our goal is to prove the following lemma.

\hittingset*

\begin{proof}
Our goal is to design an algorithm that after each update \emph{reports} a constant number of nodes that are added to or removed from its output set $ \mathcal{A} $.
Let $ t $ such that $ t = \Theta (\tfrac{n}{d} \log n) $ be an upper bound on the size of a hitting set of a graph with at most~$ n $ nodes of degree at most~$ d $ as computed with the algorithm of Lemma~\ref{lem:static-hitting-set}.

In our algorithm we will repeatedly use the following subroutine for \emph{fixing} a hitting set $ A $ on a node $ v $:
If $ v $ has degree at least $ d $ and does not have a node of $ A $ among its first $ d $ neighbors, then return an arbitrary neighbor of $ v $.
This fixing takes time $ O (d \log n) $ (if $ A $ is represented by a self-balancing binary search tree and the graph is represented by adjacency lists).
Note that if $ A $ is a hitting set of the current graph and we insert or delete an edge $ (u, v) $, then by fixing $ A $ on $ u $ and~$ v $ and adding the (at most two) returned nodes to $ A $, $ A $ will remain a hitting set.

In our dynamic hitting set algorithm, we subdivide the sequence of updates into phases of length~$ t $.
In the following, we show how to maintain a hitting set with the desired bounds throughout an arbitrary phase.
The algorithm will internally use two sets $ A_\text{old} $ and $ A_\text{new} $ represented by self-balancing binary search trees.
Throughout, it will report changes to its output set $ \mathcal{A} $ by adding nodes to or removing nodes from $ \mathcal{A} $, but it will not always explicitly keep a representation of $ \mathcal{A} $ itself.
After the end of each phase, $ A_\text{old} $ is set to $ A_\text{new} $.

The algorithm will maintain the invariant that at the beginning of each phase $ A_\text{old} $ is equal to the maintained set $ \mathcal{A} $ and has size $ O (t) = O(\tfrac{n}{d} \log n) $ and initially this will be ensured by computing a hitting set of this size in the preprocessing stage with the algorithm of Lemma~\ref{lem:static-hitting-set}.

In our algorithm, we further subdivide each phase into five subphases of length $ \frac{t}{5} $ each and proceed as follows:
\begin{itemize}
    \item During subphases~$ 1 $--$ 4 $ we fix $ A_\text{old} $ on each endpoint of an updated edge, report the at most two returned nodes as part of $ \mathcal{A} $, and add all such returned nodes to $ A_\text{old} $.
    As discussed above, the fixing takes time $ O (d \log n) $ per update.
    \item During subphases~$ 1 $--$ 3 $, we add all endpoints of updated edges to a list $ U $.
    Note that $ | U | \leq \tfrac{6}{5} t $.
    This takes constant time per update.
    \item During subphase~$ 1 $, we additionally construct a graph $ G' $ as follows: After each update we process $ \Theta(d) $ nodes of the current input graph and for each processed node copy the edges to its first (up to) $ d $ neighbors (and potentially their incident nodes) to $ G' $.
    This takes time $ O (d^2) $ per update (for graphs represented as adjacency lists) and after this subphase of length $ \frac{t}{5} = \Omega (\tfrac{n}{d}) $ all $ n $ nodes of the input graph have been processed in this manner.
    \item During subphase~$ 2 $, we additionally compute a hitting set $ A_\text{new} $ of size $ O (\tfrac{n}{d} \log n) $ on $ G' $ using the algorithm of Lemma~\ref{lem:static-hitting-set} by performing $ \Theta (d^2) $ operations of this algorithm after each update.
    \item During subphase~$ 3 $, we report after each update $ \Theta (1) $ nodes of $ A_\text{new} $ as being added to $ \mathcal{A} $.
    This takes time $ O (\log n) $ per update.
    \item During subphases~$ 4 $--$ 5 $, we fix $ A_\text{new} $ on each endpoint of an updated edge, report the at most two returned nodes as part of $ \mathcal{A} $, and add all such returned nodes to $ A_\text{new} $.
    This takes time $ O (d \log n) $ per update.
    \item During subphase~$ 4 $, we additionally after each update fix $ A_\text{new} $ on $ 6 $ nodes from $ U $ (corresponding to endpoints of historic updates), report the at most $ 6 $ returned nodes as part of $ \mathcal{A} $, and add all returned nodes to the set $ A_\text{new} $.
    This takes time $ O (d \log n) $ per update.
    Note that after this subphase of length $ \frac{t}{5} $, we have fixed $ A_\text{new} $  on all of the at most $\tfrac{6}{5} t $ nodes of $ U $.
    \item During subphase~$ 5 $, we process after each update $ \Theta (1) $ nodes of $ A_\text{old} $, remove them from $ A_\text{old} $, and report each such node as not being part of $ \mathcal{A} $ anymore unless it has in the meanwhile been added to $ A_\text{new} $.
    This takes time $ O (\log n) $ per update.
\end{itemize}
If the invariant (i.e., that at the beginning of each phase $ A_\text{old} $ is equal to the maintained set $ \mathcal{A} $ and has size $ O (t) = O(\tfrac{n}{d} \log n) $) holds, then the bounds on the running time and the recourse follow from the inline arguments.
To prove the first part of the invariant, observe that the algorithm ensures that at the end of each phase $ A_\text{new} $ will be equal to $ \mathcal{A} $ and that this set $ A_\text{new} $ will be used as the set $ A_\text{old} $ at the beginning of the next phase.
To prove the size claim in the invariant, note that $ A_\text{new} $ has size $ O (t) $ when initialized and after each update of the phase of length $ t $ at most a constant number of nodes are added to $ A_\text{new} $.
It thus follows that $ A_\text{new} $ has size $ O (t) = O (\tfrac{n}{d} \log n) $ at the end of the phase.
As this set $ A_\text{new} $ will be used as the set $ A_\text{old} $ at the beginning of the next phase, the invariant follows.
Now, by the invariant, $ \mathcal{A} $ has size $ O (t) $ at the beginning of each phase.
Since at most a constant number of nodes are added to $ \mathcal{A} $ after each update in a phase of length $ t $, we can conclude that $ \mathcal{A} $ always size $ O (t) = O (\tfrac{n}{d} \log n) $.

The algorithm correctly maintains a hittin set based on the following observations:
Since $ A_\text{old} $ is a hitting set at the beginning of the phase and is fixed on every endpoint on an update until the end of subphase~$ 4 $, $ A_\text{old} $ is a valid hitting set until the end of subphase~$ 4 $.
Since the set~$ \mathcal{A} $ reported by the algorithm always contains $ A_\text{old} $ until the end of subphase~$ 4 $, the algorithm correctly maintains a hitting set until the end of subphase~$ 4 $.
Correctness after the end of subphase~$ 4 $ is ensured as follows:
$ A_\text{new} $ is initialized to be a hitting set of $ G' $.
Note at any time (after the end of subphase~$ 1 $) for any node that $ v $ has not been the endpoint of an update, the set of edges incident on $ v $ in $ G' $ equals the set of edges of $ v $ to its first (up to) $ d $ neighbors in the current input graph.
Thus, it suffices to fix $ A_\text{new} $ on every endpoint of an edge updated since the beginning of the current phase to ensure that $ A_\text{new} $ is a hitting set on the current input graph.
This is what the algorithm does until the end of subphase~$ 4 $.
Since the algorithm continues fixing $ A_\text{new} $ on endpoints of edges updated in subphase~$ 5 $, $ A_\text{new} $ is a valid hitting set during all of subphase~$ 5 $.
Since the set~$ \mathcal{A} $ reported by the algorithm always contains $ A_\text{new} $ from the end of subphase~$ 3 $ onwards, the algorithm correctly maintains a hitting set during subphase~$ 5 $.

\end{proof}

\section{Sparse Emulator with Applications in $(1+\epsilon)$-APSP and $(1+\epsilon)$-MSSP}\label{sec:sparse emulator}
In this section we show that by maintaining a much sparser emulator, we can maintain distances from many sources efficiently. At a high-level, we first use the construction in the previous section to maintain a $(1+\epsilon, 4)$-emulator $H$, and then use a static deterministic algorithm on $H$ to obtain a $(1+\epsilon, n^{o(1)})$-emulator with size $n^{1+o(1)}$. 

\subsection{Sparse Deterministic Emulators}
We start by showing that we can maintain near-additive emulators with general stretch/size tradeoffs.
Before describing our dynamic construction, we observe that statically we can construct near-additive spanners (and hence emulators) efficiently and deterministically. For this we can use the deterministic algorithm of \cite{RTZ2005} for constructing the clusters used in the \textit{spanner construction} of \cite{TZ2006emulators}. In other words, we can derandomize the spanner construction in \cite{TZ2006emulators} and have:
\begin{lemma}[\cite{TZ2006emulators, RTZ2005}]\label{lem:det_spanners}
Given an unweighted graph $G=(V,E)$ with $m$ edges, and an integer $k >1$, there is a deterministic algorithm that constructs, for any $ 0 < \epsilon \leq 1 $, a $(1+\epsilon, \beta)$-spanner with $\beta=(1/\epsilon)^k$ and size $O(n^{1+1/k})$ in $\tilde{O}(mn^{1/k}))$ time.
\end{lemma}
Note that while with this (static) lemma we can construct a spanner (i.e., is a true subgraph of $G$), in our deterministic dynamic construction the algebraic techniques only let us maintain an emulator efficiently.
For maintaining these emulators we perform the following for a parameter $\epsilon'$ that will be set later:
\begin{itemize}
    \item Maintain a $(1+\epsilon', 4)$-emulator $H$ of size $\tilde{O}(n^{4/3})$.
    \item Turn the emulator into an unweighted graph: we replace each \textit{weighted edge} $w(e)$ of $H$ with an unweighted path of length $ w(e)$. Since the emulator we constructed only has edge weights bounded by $O(1/\epsilon')$, this will blow-up the size only by a factor of $O(1/\epsilon')$.
    \item Using \Cref{lem:det_spanners}, we statically construct a $(1+\epsilon', \beta)$-emulator $H'$ of $H$.
\end{itemize}

It is easy to see that $H'$ is now an emulator of $G$ with slightly larger additive factor. More formally,

\emulator*
\begin{proof}
It is easy to see that the update time is the maximum of the time required for running the algorithm in Lemma \ref{lem:det_spanners} statically, and the dynamic time for maintaining the $(1+\epsilon, 4)$-emulator which is given by Theorem \ref{thm:sparse_emulator}.

For every pair of nodes $s,t$, there is a path with length $(1+\epsilon')d_H(s,t) + \beta$ in $H'$. 
We also know $d_H(s,t) \leq (1+ \epsilon') d_G(s,t)+ 2$. Hence,

\[d_H'(s,t) \leq (1+\epsilon')d_H(s,t)+ \beta \leq (1+\epsilon')[(1+\epsilon')d_G(s,t)+2] +\beta \leq (1+3\epsilon')d_G(s,t) + O(\beta)\]

Hence by setting $\epsilon' =\epsilon/3$ the claim follows.
\end{proof}

One important special case of this result is when we set $k=\sqrt{\log n}$ and $ \epsilon $ is a constant. In this case we have an additive stretch of $\beta=O(1/\epsilon)^{\sqrt{\log n}}=n^{o(1)}$, and the size of the emulator is $\tilde{O}(n^{1+o(1)})$. We can obtain such a sparse emulator in $O( \frac{n^{1.407}}{\epsilon^2 \log \epsilon^{-1}})$ worst-case update time deterministically.

\subsection{Deterministic $(1+\epsilon)$-MSSP}\label{sec:MSSP}
Now using the emulator in this special setting, we do the following for maintaining multi-source distances from a set $S$ of sources: 
\begin{enumerate}
    \item At each update, after updating $H'$, statically compute $S \times V$ distances on $H'$ in $O(|S| \cdot |E(H')|)$ time. 
    \item Maintain $O(\beta)$-bounded distances between pairs in $S \times V$ on $G$. 
    \item The distance estimate $d(s,v)$ for each source $s \in S$ and node $v$ is the minimum distance estimate derived from these two steps.
\end{enumerate}

It is now easy to combine Lemma \ref{lem:sparse_det_emulator} with Lemma \ref{lem:overview:pairwise} by setting $k=\sqrt{\frac{\log_{1/\epsilon } n}{2}}$ and thus $\beta= O(1/\epsilon)^{\sqrt{2\log_{1/\epsilon} n}}$ to get the following corollary for maintaining distances from multiple-sources:

\mssp*

Using the above approach we can compute $(1+\epsilon)$-MSSP from up to $O(n^{0.52})$ sources in the same time complexity as our $(1+\epsilon)$-SSSP algorithm, namely $O(n^{1.529})$.

\subsection{Deterministic $(1+\epsilon)$-APSP}\label{sec:APSP}
We can directly use our $(1+\epsilon)$-MSSP algorithm to maintain all-pairs-shortest paths distances deterministically by setting $S=V$. But here we give an even simpler algorithm that does not involve dynamically maintaining an emulator.

We maintain a $(1+\epsilon)$-APSP data structure with worst-case update time $O(n^{2+o(1)})$ deterministically by i) Computing an emulator with size $O(n^{1+o(1)})$ \textit{statically} in each update, ii) Maintaining $n^{o(1)}$-bounded distances using the algebraic algorithms of Lemma \ref{lem:overview:pairwise}. More formally as a direct consequence of Theorem \ref{thm:MSSP} we have,

\apsp*

\section{Pairwise Bounded Distances via Algebraic Techniques}
\label{sec:algebraic}

The emulator constructions from Section \ref{sec:emulator} require the $S\times T$-distances for some dynamically changing sets $S,T \subseteq V$.
By exploiting fast matrix multiplication, existing algebraic algorithms allow for batch queries that return all $S\times T$-distances in less time than by individually querying each $st$-distance for $(s,t)\in S\times T$.
However, these batch queries are relatively slow (see \Cref{sec:batchquerycomparison} for details) and using them to maintain distances does not allow us to reach the conditional lower bounds from \cite{BrandNS19}.

That is why we here design dynamic algorithms that are able to maintain $S\times T$ pairwise distances more efficiently for dynamically changing sets $S$ and $T$.

As outlined in \Cref{sec:overview:algebraic}, dynamic shortest paths for small distances can be reduced to dynamic matrix inverse. 
Thus maintaining $S\times T$ distances reduces to maintaining the submatrix $(\mM^{-1})_{S,T}$ for some dynamically changing $n\times n$ matrix $\mM$ and dynamically changing sets $S,T\subseteq [n]$.
The aim of this section is to prove \Cref{thm:overview:low_hop}, previously stated in the overview (\Cref{sec:overview:algebraic}). Via the reduction to matrix inverse, our first step for proving \Cref{thm:overview:low_hop} is to prove the following lemma for matrix inverse:

\begin{lemma}\label{lem:algebraic:submatrix}
For all $0 \le \nu \le \mu \le 1$ 
there exists a deterministic dynamic algorithm
that maintains $(\mM^{-1})_{S,T}$ for dynamic sets $S,T\subset[n]$ of size at most $n^\mu$ 
and an $n\times n$ dynamic matrix $\mM$ 
that is promised to stay non-singular throughout all updates.
Changing any entry of $\mM$ and adding/removing an index to/from $S$ or $T$ takes 
$O(n^{\omega(1,1,\mu)-\mu} + n^{\omega(1,\mu,\nu)-\nu} + n^{\mu+\nu} + |S\times T|)$ 
operations per update.
\end{lemma}

We prove \Cref{lem:algebraic:submatrix} in \Cref{sec:algebraic:submatrix}.
In \Cref{sec:algebraic:graph} we then use \Cref{lem:algebraic:submatrix} to prove \Cref{thm:overview:low_hop}.

\subsection{Dynamic Matrix Inverse with Submatrix Maintenance}
\label{sec:algebraic:submatrix}

Our proof of \Cref{lem:algebraic:submatrix} starts by reducing the task of explicitly maintaining some submatrix $\mM^{-1}_{S,T}$ to maintaining the inverse $\mM^{-1}$ implicitly where we allow queries to partial rows $\mM^{-1}_{i,T}$ for any $i \in [n]$.
This reduction is formalized in \Cref{lem:reduction_multi}.

\reductionMulti

\begin{proof}
During initialization we compute $\mM^{-1}$ in $O(n^\omega)$ operations, so we know $\mM^{-1}_{S,T}$.
Now consider an entry update where we add some scalar $c$ to some entry $\mM_{i,j}$, i.e.~we add some $e_i e_j^\top c$ to $\mM$.
By the Sherman-Morrison-Woodbury identity of \Cref{lem:woodbury} we have
$$
(\mM+e_i e_j^\top c)^{-1} = \mM^{-1} - \frac{\mM^{-1} e_i c e_j^\top \mM^{-1}}
{1 + c e_j^\top \mM^{-1} e_i}.
$$
Note that $\mM^{-1} e_i$ and $e_j^\top \mM^{-1}$ are just $i$th column and $j$th row of $\mM^{-1}$. 
So to obtain the submatrix $((\mM+e_i e_j^\top c)^{-1})_{S,T}$ 
we just need to compute
$$
((\mM+e_i e_j c)^{-1})_{S,T} 
= 
\mM^{-1}_{S,T} - \frac{\mM^{-1}_{S,i} c \mM^{-1}_{j,T}}
{1 + c \mM^{-1}_{i,j}}.
$$
This allows us to maintain $\mM^{-1}_{S,T}$ throughout all updates 
by just querying $\mM^{-1}_{S,i}$ and $\mM^{-1}_{j,T}$ in $O(q(k,n))$ time for $k=\max(|S|,|T|)$. 
Computing the above expression then takes an additional $O(|S\times T|)$ operations. 
Note that we can support queries to $\mM^{-1}_{S,i}$ 
by running a second copy in parallel on matrix $\mM^\top$ 
and querying $(\mM^\top)^{-1}_{i,S} = (\mM^{-1}_{S,i})^\top$.
Updates to set $S$ (or $T$) are handled by querying the missing row (or column), 
i.e.~adding some $1\le k\le n$ to $S$ just requires us to query $\mM^{-1}_{k,T}$ 
in order to know the new submatrix $\mM^{-1}_{S\cup\{k\},T}$.

\end{proof}

To prove \Cref{lem:algebraic:submatrix}, we must now design a dynamic algorithm that is able to efficiently maintain the inverse $\mM^{-1}$ in implicit form that allows fast queries to any partial row $\mM^{-1}_{i,T}$.
The dynamic algorithm by \cite{BrandNS19} has the fastest update complexity, but a rather slow query complexity.
We modify this algorithm to support faster queries to partial rows $\mM^{-1}_{i,T}$ by exploiting the fact that set $T\subset[n]$ is slowly changing.

As outlined in \Cref{sec:overview:algebraic}, for any $0\le\nu\le\mu\le1$ the dynamic algorithm maintains the dynamic matrix $\mM$ in implicit form
\begin{align}
\mM = \mM' + \mU' \mV'^\top + \mU \mV^\top
\label{eq:algebraic:implicit}
\end{align}
where $\mM'$ is matrix $\mM$ at most $n^\mu$ updates ago,
and $\mU'\mV'$ have at most $n^\mu$ columns, 
and $\mU,\mV$ have at most $n^\nu$ columns.
Initially, $\mU,\mU',\mV,\mV'$ are all empty matrices (i.e. with 0 columns) as $\mM = \mM'$.
Then, with each update to $\mM$, we update $\mU$ and $\mV$ as follows:
The entry update to $\mM$ can be represented as adding some $v\cdot e_i e_j^\top$ to $\mM$. We can thus maintain \eqref{eq:algebraic:implicit} by setting $\mU \leftarrow [\mU|v\cdot e_i]$ and $\mV \leftarrow [\mV, e_j]$.
After $n^\nu$ updates, the matrices $\mU$ and $\mV$ have $n^\nu$ columns 
and we set $\mU' \leftarrow [\mU'|\mU]$, $\mV' \leftarrow [\mV'|\mV]$
and reset $\mU,\mV$ to be empty matrices.
Thus invariant \eqref{eq:algebraic:implicit} is still maintained and we can assume $\mU$ and $\mV$ always have at most $n^\nu$ columns.
After $n^\mu$ updates, the algorithm is reset by letting $\mM' \leftarrow \mM$ and all $\mU,\mU',\mV,\mV'$ are reset to be empty matrices. Thus we can also assume $\mU',\mV'$ have at most $n^\mu$ columns.

We now design two dynamic algorithms:
(i) a dynamic algorithm that maintains the inverse of $\mM'' := (\mM' + \mU' \mV'^\top)$ where $\mU'$ and $\mV'$ receives $n^\nu$ new columns with each update.
(ii) a dynamic algorithm that maintains the inverse of $\mM = \mM'' + \mU \mV^\top$ where $\mU$ and $\mV$ receive only one column per update. 

Note that to apply the Sherman-Morrison-Woodbury identity (\Cref{lem:woodbury}), data structure (ii) requires access to the inverse of $\mM''$.
This inverse is maintained by data structure (i), i.e.~we run data structure (i) internally in data structure (ii).
The following \Cref{lem:batch_update} is our implementation of data structure (i).

\begin{lemma}\label{lem:batch_update}
For all $0 \le \nu \le \mu \le 1$ 
there exists a deterministic dynamic algorithm
that initializes in $O(n^\omega)$ operations on a dynamic set $T\subset[n]$
and a dynamic $n\times n$ matrix $\mM$ that is promised to stay non-singular.
The algorithm supports batch-updates to $n^\nu$ entries of $\mM$
in $O(n^{\omega(1,\mu,\nu)}+n^{\omega(1,1,\mu)-\mu+\nu})$ operations per update
and querying any entry $\mM^{-1}_{i,j}$ in $O(n^\mu)$ operations.
It also supports adding/removing an index to/from $T$ in $O(n^{\omega(1,\mu,\nu)-\nu})$ operations,
and querying any partial row $\mM^{-1}_{i,T}$ for any $i\in[n]$ in $O(n^{\nu+\mu})$ operations.
\end{lemma}

\begin{proof}
Let $\mM'$ be the matrix $\mM$ during initialization.
Then we maintain the inverse $\mM^{-1}$ in the implicit form
\begin{align}
\mM^{-1} = \mM'^{-1} + \mU'\mV'^\top \label{eq:batch_form}
\end{align}
for two matrices $\mU'$ and $\mV'$ with at most $n^\mu$ columns.
The set $T\subset[n]$ is maintained in the form $T=T^\init\cup T^\dyn$ 
where $|T^\dyn|\le n^\nu$, $|T^\init|\le n^\mu$
and we maintain entries $\mM'^{-1}_{\cdot,T^\init}$ 
(i.e.~columns with index in $T^\init$)
explicitly.
\paragraph{Initialization}
Given $\mM$ and $T\subset[n]$, we compute $\mM^{-1}$ in $O(n^{\omega})$ operations and set $T^\init \leftarrow T$, $T^\dyn\leftarrow\emptyset$, $\mU'=\mV'\leftarrow 0$.
\paragraph{Batch-update}
Consider a batch-update that changes $n^\nu$ entries of $\mM$. We can phrase these changes as adding an outer product $\mU\mV^\top$ to $\mM$ where both $\mU$ and $\mV$ consist of $n^\nu$ columns with only one nonzero entry each.
Applying \Cref{lem:woodbury} yields
$$
(\mM + \mU\mV^\top)^{-1} = \mM^{-1} - \underbrace{\mM^{-1}\mU}_{=:\mQ}(\mI+\underbrace{\mV^\top\mM^{-1}\mU}_{=:\mR})^{-1}\underbrace{\mV^\top\mM^{-1}}_{=:\mS}
$$
Since $\mU$ and $\mV$ have only one nonzero entry per column, 
the matrix $\mQ$ consists of $n^\nu$ columns of $\mM^{-1}$,
$\mS$ of $n^\nu$ rows of $\mM^{-1}$,
and $\mR$ is an $n^\nu\times n^\nu$ submatrix of $\mM^{-1}$.
(Though some rows/columns are rescaled since the nonzero entries of $\mU$ and $\mV$ are not necessarily $1$.)
Computing these rows and columns of $\mM^{-1}$ takes $O(n^{\omega(1,\mu,\nu)})$ operations by \eqref{eq:batch_form} and both $\mU'$, $\mV'$ having at most $n^\mu$ columns.
Computing $(\mI+\mR)^{-1}$ and multiplying with $\mS$ then takes 
$O(n^{\omega\cdot\nu}+n^{\omega(1,\nu,\nu)}) = O(n^{\omega(1,\nu,\nu)})$
operations.
We can append the resulting matrix to $\mV'$ and append $\mQ$ to $\mU'$ to obtain
\begin{align*}
&~
(\mM + \mU\mV^\top)^{-1} 
=
\mM^{-1} - \mQ (\mI+\mR)^{-1} \mS \\
=&~ 
\mM'^{-1} - \mU'\mV'^\top - \mQ (\mI+\mR)^{-1} \mS 
=
\mM'^{-1} - \underbrace{[\mU'|\mQ]}_{\text{new }\mU'} 
\underbrace{
\left[
\begin{array}{c}
\mV'^\top \\
\hline
(\mI+\mR)^{-1} \mS
\end{array}
\right]
}_{\text{new }\mV'^\top}.
\end{align*}
Thus, we can maintain the inverse in form \eqref{eq:batch_form}.
Additionally, to maintain the entries $\mM^{-1}_{[n],T^\init}$ explicitly we just need to add $-(\mQ(\mI+\mR)\mS)_{[n],T^\init}$ to it, which takes $O(n^{\omega(1,\nu,\mu)})$ operations by $|T^\init| \le n^\mu$.

As soon as more than $n^\mu$ entries were changed (i.e.~after $n^{\mu-\nu}$ batch-updates) we restart the algorithm but instead of computing $\mM^{-1}$ in $O(n^\omega)$ operations we compute $\mM^{-1} = \mM'^{-1} + \mU'\mV'^\top$ in $O(n^{\omega(1,1,\mu)})$ operations.
Thus, the update complexity is $O(n^{\omega(1,1,\mu)-\mu+\nu}+n^{\omega(1,\mu,\nu)})$ amortized operations per update, which can be made worst-case via standard techniques (see e.g.~\cite[Theorem B.1]{BrandNS19}).
\paragraph{Set-update}
If an index $i \in [n]$ is added to $T$, we add the index to $T^\dyn$.
If $|T^\dyn|>n^\nu$, we set $T^\init \leftarrow T^\init \cup T^\dyn$ and $T^\dyn \leftarrow \emptyset$.
To maintain the invariant that the entries $\mM^{-1}_{\cdot,T^\init}$ are always maintained explicitly, we compute $\mM^{-1}_{[n],T^\dyn}$ in $O(n^{\omega(1,\mu,\nu)})$ operations by taking the product $\mU'\mV'^\top_{[n^\mu],T^\dyn}$.
The amortized update time is thus $O(n^{\omega(1,\mu,\nu)-\nu})$ which can be made worst-case via standard techniques.
When removing an index from $T$, we remove the index from $T^\init$ and $T^\dyn$. 
\paragraph{Query}
To query any entry $\mM^{-1}_{i,j}$, we compute $e_i^\top\mU'\mV'^\top e_j$ which takes $O(n^\mu)$ operations.
To compute $\mM^{-1}_{i,T}$ for any $i\in[n]$ we query each entry of $\mM^{-1}_{i,T^\dyn}$ individually as $\mM'^{-1}_{i,T^\init}$ is already known.
This takes $O(n^{\mu+\nu})$ operations.
\end{proof}

By maintaining $\mM^{-1} = (\mM'' + \mU \mV^\top)^{-1}$ where $\mM''^{-1}$ is maintained via the previous \Cref{lem:batch_update}, we obtain the following dynamic matrix inverse algorithm.
\Cref{lem:algebraic:entry_update} together with \Cref{lem:reduction_multi} then implies \Cref{lem:algebraic:submatrix}.

\begin{lemma}\label{lem:algebraic:entry_update}
For any $0 \le \nu \le \mu \le 1$ 
there exists a deterministic dynamic algorithm
that initializes in $O(n^\omega)$ operations
on a dynamic set $T\subset[n]$ 
and dynamic $n\times n$ matrix $\mM$ that is promised to stay non-singular. 
The algorithm supports entry updates to $\mM$ and adding/removing an index to/from $T$ in $O(n^{\omega(1,1,\mu)-\mu} + n^{\omega(1,\mu,\nu)-\nu} + n^{\mu+\nu})$ operations per update.
Querying any partial row $\mM^{-1}_{i,T}$ for any $i\in[n]$ takes $O(n^{\mu+\nu})$ operations.
\end{lemma}

\begin{proof}
We first describe in which form we maintain matrix $\mM$ and its inverse $\mM^{-1}$.
We consider some matrix $\mM'$ which is the matrix $\mM$ some time in the past, but at most $n^\nu$ updates ago.
Thus we can write $\mM = \mM' + \mU\mV^\top$ for two $n\times n^\nu$ matrices $\mU$ and $\mV$ with at most one nonzero entry per column.
We consider the size of $\mU$ and $\mV$ fixed, but they may contain columns that are all-zero. With each entry update to $\mM$, one all-zero column of $\mU$ and $\mV$ obtain one non-zero entry.

By the Sherman-Morrison-Woodbury identity of \Cref{lem:woodbury}, 
we can maintain the inverse $\mM^{-1}$ implicitly 
via an $n^\nu \times n^\nu$ matrix $\mN$ 
where
\begin{align}
\mM^{-1} 
= 
(\mM' + \mU\mV^\top)^{-1} 
= 
\mM'^{-1} - 
\mM'^{-1}\mU
\underbrace{(\mI+\mV^\top\mM'^{-1}\mU)^{-1}}_{=:\mN}
\mV^\top\mM'^{-1}\label{eq:batch_form2}
\end{align}
Here we maintain $\mM'^{-1}$ via the dynamic algorithm of \Cref{lem:batch_update}.
We will reset $\mU=\mV=0$ and $\mN=\mI$ every $n^\nu$ updates and perform the past $n^\nu$ updates to $\mM'$ (via \Cref{lem:batch_update}).
This way we are able to guarantee the dimensionality bound on the matrices $\mU,\mV,\mN$.
\paragraph{Initialization}
Given $\mM$ and $T\subset[n]$, we initialize the dynamic algorithm of \Cref{lem:batch_update} on $\mM$ in $O(n^{\omega})$ operations.
\paragraph{Entry-update}
Consider an entry update to $\mM$, then one all-zero column of $\mU$ and $\mV$ obtains a non-zero entry.
To maintain the matrix $\mN = (\mI+\mV^\top\mM'^{-1}\mU)^{-1}$,
note that $\mV^\top\mM'^{-1}\mU$ changes in only one row and column per update.
Further, these rows and columns contain $O(n^\nu)$ entries of $\mM'^{-1}$.
Thus we can maintain $\mV^\top\mM'^{-1}\mU$ by doing the following:

(i) Perform $O(n^\nu)$ queries to the dynamic algorithm that maintains $\mM'^{-1}$ (\Cref{lem:batch_update}), which takes $O(n^{\nu+\mu})$ operations,

(ii) Update $\mN$ by using \Cref{lem:woodbury} on $(\mI+\mV^\top\mM'^{-1}\mU)$.
With each update, this matrix changes in only one row and column, so we can maintain its inverse (i.e.~$\mN$) via \Cref{lem:woodbury} in $O(n^{2\nu})$ operations.

We further explicitly maintain the entries $(\mV^\top\mM'^{-1})_{[n^\nu],T}$.
Changing one all-zero column of $\mV$ by adding a single nonzero entry to that column,
will change $\mV^\top\mM'^{-1}$ in only one row.
Further, this changed row contains a row of $\mM'^{-1}$.
Thus to maintain $(\mV^\top\mM'^{-1})_{[n^\nu],T}$, we must only query $\mM'^{-1}_{i,T}$ for some $i\in[n]$ which takes $O(n^{\nu+\mu})$ time by \Cref{lem:batch_update}.

The update complexity is thus 
$O(n^{\mu+\nu} + n^{\omega(1,1,\mu)-\mu} + n^{\omega(1,\mu,\nu)-\nu})$ 
where the the last two terms come from updating \Cref{lem:batch_update} every $n^\nu$ updates.
This amortized update complexity can be made worst-case via standard techniques (see e.g.~\cite[Theorem B.1]{BrandNS19}).
\paragraph{Set-update}
When an index is added/removed to/from set $T$, we pass this update to the data structure of \Cref{lem:batch_update}.
That dynamic algorithm then performs $O(n^{\omega(1,\mu,\nu)-\nu})$ field operations.
Further, if some index $i\in[n]$ is added to $T$, we must compute $(\mV^\top\mM'^{-1})_{[n^\nu],i}$ in order to maintain $(\mV^\top\mM'^{-1})_{[n^\nu],T}$ explicitly.
As $\mV$ contains at most one non-zero entry per column and has $n^\nu$ columns, this corresponds to querying $n^\nu$ entries of $\mM'$, which takes $O(n^{\mu+\nu})$ operations (\Cref{lem:batch_update}).
A set-update thus takes $O(n^{\mu+\nu}+n^{\omega(1,\mu,\nu)-\nu})$ operations.
\paragraph{Queries}
To query any entry $\mM^{-1}_{i,j}$ we first query entry $\mM'^{-1}_{i,j}$ via \Cref{lem:batch_update} in $O(n^\mu)$ operations.
Further, we compute $(e_i^\top \mM'^{-1}\mU) \mN (\mV^\top \mM'^{-1} e_j)$, 
where the right-most and left-most term are just $O(n^\nu)$ entries of $\mM'^{-1}$
(as $\mU$ and $\mV$ contain at most one non-zero entry per column).
Obtaining these entries of $\mM'^{-1}$ takes $O(n^{\nu+\mu})$ operations (\Cref{lem:batch_update}).
The product with $\mN$ takes an additional $O(n^{2\nu})$ operations which is subsumed by the previous cost.
The query complexity is thus $O(n^{\nu+\mu})$ operations.

To query $(\mM^{-1})_{i,T}$ for any $i\in[n]$, we can query $\mM'^{-1}_{i,T}$ via \Cref{lem:batch_update} in $O(n^{\nu+\mu})$ operations.
We must also compute $(\mM'^{-1}\mU \mN \mV^\top \mM'^{-1})_{i,T} = (e_i^\top \mM'^{-1} \mU) \mN (\mV^\top \mM^{-1})_{[n^\nu],T}$.
Here $(e_i^\top \mM'^{-1} \mU)$ are just $O(n^\nu)$ entries of $\mM'^{-1}$ which are obtained in $O(n^{\nu+\mu})$ via \Cref{lem:batch_update}.
The entries of $(\mV^\top \mM^{-1})_{[n^\nu],T}$ are maintained in the update routine described above.
Multiplying all these terms takes $O(n^{2\nu}+n^{\nu}|T|)=O(n^{\nu+\mu})$ operations.
In summary, answering a query to $(\mM^{-1})_{i,T}$ for any $i\in[n]$ takes $O(n^{\nu+\mu})$ operations.
\end{proof}

At last, we observe that the algorithms presented here work also for a polynomial matrix of the form $(\mI-X\cdot\mA) \in (\F[X]\langle X^h \rangle)^{n\times n}$ under entry updates to $\mA$. 
When $\mA$ is the adjacency matrix of some graph, the inverse of such a matrix encodes the distances in $G$. The respective reduction was previously outlined in \Cref{sec:overview:algebraic} and will be used formally in \Cref{sec:algebraic:graph}.

\begin{corollary}\label{cor:algebraic:polynomial}
\Cref{lem:algebraic:submatrix} can also maintain the inverse of $(\mI-X\mA) \in (\F[X]/\langle X^h \rangle)^{n \times n}$ under entry updates to $\mA$.
The update and query complexities increase by a factor of $O(h \log h)$ field operations.
\end{corollary}

\begin{proof}
The statements of \Cref{lem:algebraic:submatrix,lem:reduction_multi,lem:batch_update,lem:algebraic:entry_update} assume that the input matrix is over some field $\F$ and measures the complexity in the number of field operations in $\F$.
We now consider what happens if the matrix has entries from some ring $\F[X]/\langle X^h \rangle$ instead (i.e.~polynomials with degrees truncated to $<h$).
The algorithms of \Cref{lem:algebraic:submatrix,lem:reduction_multi,lem:batch_update,lem:algebraic:entry_update} only perform matrix operations (i.e.~matrix products or matrix inversions). 
To verify that the algorithms work for polynomial matrices, we must only verify that the matrices the algorithms try to invert are invertible.
For this, note that all matrix-inversions come from the Sherman-Morrison-Woodbury identity \Cref{lem:woodbury} which attempts to invert some matrix of the form $(\mI + \mV^\top \mM^{-1} \mU)$.
Here all entries of $\mU$ can be assumed to be a multiple of $X$, because we maintain the inverse of $\mI-X\cdot\mA$ with entry updates to $\mA$, i.e.~updates of form $\mI-X\mA+Xve_i e_j^\top$ for $v\in\F[X]\langle X^h\rangle$.
Thus $\mU$ is a multiple of $X$ and the inversion of $(\mI + \mV^\top \mM^{-1} \mU)$ is well-defined.
(The inverse of a matrix of the form $\mI + X\mN$ is $\sum_{k=0}^h X^k (-\mN)^k$.)
Finally, note that there is a complexity blow-up of a factor $O(h \log h)$ because each matrix entry is a polynomial of degree at most $h$.

\end{proof}

\subsection{Pairwise bounded distances}
\label{sec:algebraic:graph}

In \Cref{sec:algebraic:submatrix} we have shown how to maintain the submatrix of some dynamic matrix inverse.
We now use this result to formally prove \Cref{thm:overview:low_hop} which we restate here for convenience:

\lowHopThm*

To maintain $st$-distances, we will pick $|S|,|T| = \tilde{O}(n^{2/3})$, resulting in an $O(n^{1.407}/\epsilon^2 \log \epsilon^{-1})$ time algorithm. For the single-source case we will pick $|S|=n$ and $|T| = \tilde{O}(\sqrt{n})$, resulting in an $O(n^{1.529}/\epsilon^2 \log \epsilon^{-1})$ time algorithm.



\begin{proof}
Given adjacency matrix $\mA$ of a dynamic graph
we consider the inverse of $\mI-X\mA \in (\F[X]/\langle X^{h+1} \rangle)^{n\times n}$. We work with the field $\F = \Z_p$ for some prime $n^h < p \le 2n^h$.
We will argue that this $p$ is large enough that our numbers never wrap around (i.e. are always smaller than $p$). Specifically, any zeroes we encounter in our results will be true zeroes and not some multiples of $p$ that became $0$ by the modulo operation.

Note that $\mA_{u,v}^k$ is exactly the number of (not necessarily simple) paths from $u$ to $v$ of distance $k$.
For any $u,v$, the total number of $uv$-paths of distance $k$ is at most $n^k <p$ when $k \le h$, so by $(\mI-X\mA)^{-1} = \sum_{k=0}^h X^k\mA^k$, we can read any $uv$-distance $\le h$ by looking at the smallest degree with non-zero coefficient in $(\mI-X\mA)^{-1}_{u,v}$.
So to maintain pairwise distances of $S\times T$, we need to maintain the submatrix $(\mI-X\mA)^{-1}_{S,T}$ under entry-updates to $\mA$ and updates to $S$ and $T$.
This is done via \Cref{lem:algebraic:submatrix}.

At last, we observe that each field operation takes $O(h)$ time in the Word-RAM model because they fit in $O(h)$ words \cite{Knuth97}, so the complexities increase by a factor of $O(h)$, which leads to the complexities stated in \Cref{thm:overview:low_hop}.
\end{proof}

We remark that, if we are fine with randomization, we can save a factor $h$ in the complexity of \Cref{thm:overview:low_hop}. For example, one could pick a random prime $p$ with bitlength $O(\log n)$ to use for field $\Z_p$. Then for any $u,v \in V$ we have with constant probability than the number of shortest $uv$-paths is nonzero modulo $p$. So, running $O(\log n)$ copies of the data structure in parallel with independent random $p$, allows to maintain the distances w.h.p.

Given the short bit-length of the prime $p$, one arithmetic operation takes only $O(1)$ time in Word-RAM model and we obtain the following \Cref{thm:randomized:low_hop}.

\begin{corollary}
    \label{thm:randomized:low_hop}
    For all $0\le\nu\le\mu\le1$ there exists a randomized dynamic algorithm that, after preprocessing a given unweighted directed graph $G$ and sets $S,T \subseteq V$, supports edge-updates to $G$ and set-updates to $S$ and $T$ (i.e.~adding or removing a node to $S$ or $T$) as long as $|S|,|T|\le n^\mu$ throughout all updates.
    After each edge- or set-update the algorithm returns the $h$-bounded pairwise distances of $S\times T$ in $G$.
   
    The preprocessing time is $\tilde{O}(n^\omega h)$,
    and the worst-case update time is
    $$\tilde{O}((n^{\omega(1,1,\mu)-\mu} + n^{\omega(1,\mu,\nu)-\nu} + n^{\mu+\nu} + |S\times T|) h).$$
   
    The returned distances are correct with high probability.
\end{corollary}

\section*{Acknowledgement}
Jan van den Brand is funded by ONR BRC grant N00014-18-1-2562 and by the Simons Institute for the Theory of Computing through a Simons-Berkeley Postdoctoral Fellowship. Sebastian Forster and Yasamin Nazari are supported by the Austrian Science Fund (FWF): P 32863-N. This project has received funding from the European Research Council (ERC) under the European Union's Horizon 2020 research and innovation programme (grant agreement No 947702).

\printbibliography[heading=bibintoc]

@book{Woodbury50,
  title={Inverting modified matrices},
  author={Woodbury, Max A},
  year={1950},
  publisher={Statistical Research Group}
}

@inproceedings{Sankowski07,
  author    = {Piotr Sankowski},
  title     = {Faster dynamic matchings and vertex connectivity},
  booktitle = {{SODA}},
  pages     = {118--126},
  publisher = {{SIAM}},
  year      = {2007}
}

@inproceedings{ChoudharyG20,
  author    = {Keerti Choudhary and
               Omer Gold},
  title     = {Extremal Distances in Directed Graphs: Tight Spanners and Near-Optimal
               Approximation Algorithms},
  booktitle = {{SODA}},
  pages     = {495--514},
  publisher = {{SIAM}},
  year      = {2020}
}

@inproceedings{AnconaHRWW19,
  author    = {Bertie Ancona and
               Monika Henzinger and
               Liam Roditty and
               Virginia Vassilevska Williams and
               Nicole Wein},
  title     = {Algorithms and Hardness for Diameter in Dynamic Graphs},
  booktitle = {{ICALP}},
  series    = {LIPIcs},
  volume    = {132},
  pages     = {13:1--13:14},
  publisher = {Schloss Dagstuhl - Leibniz-Zentrum f{\"{u}}r Informatik},
  year      = {2019}
}

@inproceedings{Sankowski08,
  author    = {Piotr Sankowski},
  title     = {Algebraic Graph Algorithms},
  booktitle = {{MFCS}},
  series    = {Lecture Notes in Computer Science},
  volume    = {5162},
  pages     = {68--82},
  publisher = {Springer},
  year      = {2008}
}

@inproceedings{KarczmarzMS22,
  author    = {Adam Karczmarz and Anish Mukherjee and Piotr Sankowski},
  title     = {Subquadratic Dynamic Path Reporting in Directed Graphs Against an Adaptive Adversary},
  booktitle = {{STOC}},
  publisher = {{ACM}},
  year      = {2022}
}

@inproceedings{AlmanW21,
  author    = {Josh Alman and
               Virginia Vassilevska Williams},
  title     = {A Refined Laser Method and Faster Matrix Multiplication},
  booktitle = {{SODA}},
  pages     = {522--539},
  publisher = {{SIAM}},
  year      = {2021}
}

@inproceedings{Williams12,
  author    = {Virginia Vassilevska Williams},
  title     = {Multiplying matrices faster than coppersmith-winograd},
  booktitle = {{STOC}},
  pages     = {887--898},
  publisher = {{ACM}},
  year      = {2012}
}

@inproceedings{Gall14,
  author    = {Fran{\c{c}}ois Le Gall},
  title     = {Powers of tensors and fast matrix multiplication},
  booktitle = {{ISSAC}},
  pages     = {296--303},
  publisher = {{ACM}},
  year      = {2014}
}

@inproceedings{GallU18,
  author    = {Francois Le Gall and
               Florent Urrutia},
  title     = {Improved Rectangular Matrix Multiplication using Powers of the Coppersmith-Winograd
               Tensor},
  booktitle = {{SODA}},
  pages     = {1029--1046},
  publisher = {{SIAM}},
  year      = {2018}
}

@book{Knuth97,
   title =     {The art of computer programming. Vol.2. Seminumerical algorithms},
   author =    {Knuth, Donald E},
   publisher = {Addison-Wesley Professional},
   year =      {1997},
}

@inproceedings{DemetrescuI00,
  author    = {Camil Demetrescu and
               Giuseppe F. Italiano},
  title     = {Fully Dynamic Transitive Closure: Breaking Through the $O(n^2)$
               Barrier},
  booktitle = {{FOCS}},
  pages     = {381--389},
  publisher = {{IEEE} Computer Society},
  year      = {2000}
}

@article{ShermanM50,
  title={Adjustment of an inverse matrix corresponding to a change in one element of a given matrix},
  author={Sherman, Jack and Morrison, Winifred J},
  journal={The Annals of Mathematical Statistics},
  volume={21},
  number={1},
  pages={124--127},
  year={1950},
  publisher={JSTOR}
}

@inproceedings{BrandS19,
  author    = {Brand, Jan van den and
               Thatchaphol Saranurak},
  title     = {Sensitive Distance and Reachability Oracles for Large Batch Updates},
  booktitle = {{FOCS}},
  pages     = {424--435},
  publisher = {{IEEE} Computer Society},
  year      = {2019}
}

@inproceedings{BrandN19,
  author    = {Brand, Jan van den and
               Danupon Nanongkai},
  title     = {Dynamic Approximate Shortest Paths and Beyond: Subquadratic and Worst-Case
               Update Time},
  booktitle = {{FOCS}},
  pages     = {436--455},
  publisher = {{IEEE} Computer Society},
  year      = {2019}
}

@inproceedings{BrandNS19,
  author    = {Brand, Jan van den and
               Danupon Nanongkai and
               Thatchaphol Saranurak},
  title     = {Dynamic Matrix Inverse: Improved Algorithms and Matching Conditional
               Lower Bounds},
  booktitle = {{FOCS}},
  pages     = {456--480},
  publisher = {{IEEE} Computer Society},
  year      = {2019}
}

@inproceedings{Sankowski05,
  author    = {Piotr Sankowski},
  title     = {Subquadratic Algorithm for Dynamic Shortest Distances},
  booktitle = {{COCOON}},
  series    = {Lecture Notes in Computer Science},
  volume    = {3595},
  pages     = {461--470},
  publisher = {Springer},
  year      = {2005}
}

@inproceedings{Sankowski04,
  author    = {Piotr Sankowski},
  title     = {Dynamic Transitive Closure via Dynamic Matrix Inverse (Extended Abstract)},
  booktitle = {{FOCS}},
  pages     = {509--517},
  publisher = {{IEEE} Computer Society},
  year      = {2004}
}

@inproceedings{RTZ2005,
  title={Deterministic constructions of approximate distance oracles and spanners},
  author={Roditty, Liam and Thorup, Mikkel and Zwick, Uri},
  booktitle={International Colloquium on Automata, Languages, and Programming},
  pages={261--272},
  year={2005},
  organization={Springer}
}

@article{HKN2013,
  author    = {Monika Henzinger and
               Sebastian Krinninger and
               Danupon Nanongkai},
  title     = {Dynamic Approximate All-Pairs Shortest Paths: Breaking the O(mn) Barrier
               and Derandomization},
  journal   = {{SIAM} Journal on Computing},
  volume    = {45},
  number    = {3},
  pages     = {947--1006},
  year      = {2016},
  doi       = {10.1137/140957299},
  note      = {Announced at FOCS 2013.}
}

@inproceedings{BHGWW2021,
  title={New techniques and fine-grained hardness for dynamic near-additive spanners},
  author={Bergamaschi, Thiago and Henzinger, Monika and Gutenberg, Maximilian Probst and Williams, Virginia Vassilevska and Wein, Nicole},
  booktitle={Proceedings of the 2021 ACM-SIAM Symposium on Discrete Algorithms (SODA)},
  pages={1836--1855},
  year={2021},
  organization={SIAM}
}

@article{johnson1974,
  title={Approximation algorithms for combinatorial problems},
  author={Johnson, David S},
  journal={Journal of computer and system sciences},
  volume={9},
  number={3},
  pages={256--278},
  year={1974},
  publisher={Elsevier}
}

@inproceedings{TZ2006emulators,
author = {Thorup, Mikkel and Zwick, Uri},
title = {Spanners and Emulators with Sublinear Distance Errors},
year = {2006},
publisher = {Society for Industrial and Applied Mathematics},
booktitle = {Proceedings of the Seventeenth Annual ACM-SIAM Symposium on Discrete Algorithm},
pages = {802–809},
numpages = {8},
location = {Miami, Florida},
series = {SODA '06}
}

@inproceedings{BHK2017,
  title={Improved algorithms for computing the cycle of minimum cost-to-time ratio in directed graphs},
  author={Bringmann, Karl and Hansen, Thomas Dueholm and Krinninger, Sebastian},
  booktitle={44th International Colloquium on Automata, Languages, and Programming, ICALP 2017},
  pages={124},
  year={2017}
}

@inproceedings{GuR21,
  author    = {Yong Gu and
               Hanlin Ren},
  title     = {Constructing a Distance Sensitivity Oracle in $O(n^{2.5794}M)$
               Time},
  booktitle = {{ICALP}},
  series    = {LIPIcs},
  volume    = {198},
  pages     = {76:1--76:20},
  publisher = {Schloss Dagstuhl - Leibniz-Zentrum f{\"{u}}r Informatik},
  year      = {2021}
}

@article{DemetrescuI04,
  author    = {Camil Demetrescu and
               Giuseppe F. Italiano},
  title     = {A new approach to dynamic all pairs shortest paths},
  journal   = {Journal of the {ACM}},
  volume    = {51},
  number    = {6},
  pages     = {968--992},
  year      = {2004},
  doi       = {10.1145/1039488.1039492},
  note      = {Announced at {STOC} 2003}
}

@inproceedings{Thorup04,
  author    = {Mikkel Thorup},
  editor    = {Torben Hagerup and
               Jyrki Katajainen},
  title     = {Fully-Dynamic All-Pairs Shortest Paths: Faster and Allowing Negative
               Cycles},
  booktitle = {Proceedings of the 9th Scandinavian Workshop on Algorithm ({SWAT} 2004)},
  pages     = {384--396},
  year      = {2004},
  doi       = {10.1007/978-3-540-27810-8\_33}
}

@inproceedings{BhattacharyaHN17,
  author    = {Sayan Bhattacharya and
               Monika Henzinger and
               Danupon Nanongkai},
  title     = {Fully Dynamic Approximate Maximum Matching and Minimum Vertex Cover
               in $ O (\log^3 n) $ Worst Case Update Time},
  booktitle = {Proceedings of the Twenty-Eighth Annual {ACM-SIAM} Symposium on Discrete
               Algorithms ({SODA} 2017)},
  pages     = {470--489},
  year      = {2017},
  doi       = {10.1137/1.9781611974782.30}
}

@inproceedings{BhattacharyaCHN18,
  author    = {Sayan Bhattacharya and
               Deeparnab Chakrabarty and
               Monika Henzinger and
               Danupon Nanongkai},
  title     = {Dynamic Algorithms for Graph Coloring},
  booktitle = {Proceedings of the Twenty-Ninth Annual {ACM-SIAM} Symposium on Discrete
               Algorithms ({SODA} 2018)},
  pages     = {1--20},
  year      = {2018},
  doi       = {10.1137/1.9781611975031.1}
}

@inproceedings{SawlaniW20,
  author    = {Saurabh Sawlani and
               Junxing Wang},
  title     = {Near-optimal fully dynamic densest subgraph},
  booktitle = {Proccedings of the 52nd Annual {ACM} {SIGACT} Symposium on Theory
               of Computing ({STOC} 2020)},
  pages     = {181--193},
  year      = {2020},
  doi       = {10.1145/3357713.3384327}
}

@inproceedings{ChuzhoyGLNPS20,
  author    = {Julia Chuzhoy and
               Yu Gao and
               Jason Li and
               Danupon Nanongkai and
               Richard Peng and
               Thatchaphol Saranurak},
  title     = {A Deterministic Algorithm for Balanced Cut with Applications to Dynamic
               Connectivity, Flows, and Beyond},
  booktitle = {Proceedings of the 61st {IEEE} Annual Symposium on Foundations of Computer Science ({FOCS} 2020)},
  pages     = {1158--1167},
  year      = {2020},
  doi       = {10.1109/FOCS46700.2020.00111}
}

@article{RZ12,
  title={Dynamic approximate all-pairs shortest paths in undirected graphs},
  author={Roditty, Liam and Zwick, Uri},
  journal={SIAM Journal on Computing},
  volume={41},
  number={3},
  pages={670--683},
  year={2012},
  publisher={SIAM}
}

@article{KingS02,
  author    = {Valerie King and
               Garry Sagert},
  title     = {A Fully Dynamic Algorithm for Maintaining the Transitive Closure},
  journal   = {Journal of Computer and System Sciences},
  volume    = {65},
  number    = {1},
  pages     = {150--167},
  year      = {2002},
  doi       = {10.1006/jcss.2002.1883},
  note      = {Announced at {STOC} 1999}
}

@article{DorHZ00,
  author    = {Dorit Dor and
               Shay Halperin and
               Uri Zwick},
  title     = {All-Pairs Almost Shortest Paths},
  journal   = {{SIAM} Journal on Computing},
  volume    = {29},
  number    = {5},
  pages     = {1740--1759},
  year      = {2000},
  doi       = {10.1137/S0097539797327908},
  note      = {Announced at {FOCS} 1996}
}

@inproceedings{RV13,
  title={Fast approximation algorithms for the diameter and radius of sparse graphs},
  author={Roditty, Liam and Vassilevska Williams, Virginia},
  booktitle={Proceedings of the forty-fifth annual ACM symposium on Theory of computing},
  pages={515--524},
  year={2013}
}

@inproceedings{JinS21,
  author    = {Wenyu Jin and
               Xiaorui Sun},
  title     = {Fully Dynamic $s$-$t$ Edge Connectivity in Subpolynomial Time},
  booktitle = {Proceedings of the 62nd {IEEE} Annual Symposium on Foundations of Computer Science ({FOCS} 2021)},
  year      = {2021},
  note      = {To appear}
}

@article{UllmanY91,
  author    = {Jeffrey D. Ullman and
               Mihalis Yannakakis},
  title     = {High-Probability Parallel Transitive-Closure Algorithms},
  journal   = {{SIAM} Journal on Computing},
  volume    = {20},
  number    = {1},
  pages     = {100--125},
  year      = {1991},
  doi       = {10.1137/0220006},
  note      = {Announced at {SPAA} 1990}
}

@inproceedings{gupta2017,
  title={Online and dynamic algorithms for set cover},
  author={Gupta, Anupam and Krishnaswamy, Ravishankar and Kumar, Amit and Panigrahi, Debmalya},
  booktitle={Proceedings of the 49th Annual ACM SIGACT Symposium on Theory of Computing},
  pages={537--550},
  year={2017}
}

@inproceedings{abboud2019,
  title={Dynamic set cover: improved algorithms and lower bounds},
  author={Abboud, Amir and Addanki, Raghavendra and Grandoni, Fabrizio and Panigrahi, Debmalya and Saha, Barna},
  booktitle={Proceedings of the 51st Annual ACM SIGACT Symposium on Theory of Computing},
  pages={114--125},
  year={2019}
}

@inproceedings{BHN2019,
  title={A new deterministic algorithm for dynamic set cover},
  author={Bhattacharya, Sayan and Henzinger, Monika and Nanongkai, Danupon},
  booktitle={2019 IEEE 60th Annual Symposium on Foundations of Computer Science (FOCS)},
  pages={406--423},
  year={2019},
  organization={IEEE}
}

@inproceedings{BHNW2021,
  title={Dynamic set cover: Improved amortized and worst-case update time},
  author={Bhattacharya, Sayan and Henzinger, Monika and Nanongkai, Danupon and Wu, Xiaowei},
  booktitle={Proceedings of the 2021 ACM-SIAM Symposium on Discrete Algorithms (SODA)},
  pages={2537--2549},
  year={2021},
  organization={SIAM}
}

@article{EN2018,
  title={Efficient algorithms for constructing very sparse spanners and emulators},
  author={Elkin, Michael and Neiman, Ofer},
  journal={ACM Transactions on Algorithms (TALG)},
  volume={15},
  number={1},
  pages={1--29},
  year={2018},
  publisher={ACM New York, NY, USA}
}

@article{Ben-DavidBKTW94,
  author    = {Shai Ben{-}David and
               Allan Borodin and
               Richard M. Karp and
               G{\'{a}}bor Tardos and
               Avi Wigderson},
  title     = {On the Power of Randomization in On-Line Algorithms},
  journal   = {Algorithmica},
  volume    = {11},
  number    = {1},
  pages     = {2--14},
  year      = {1994},
  doi       = {10.1007/BF01294260},
  note      = {Announced at {STOC} 1990}
}

@article{ElkinP04,
  author    = {Michael Elkin and
               David Peleg},
  title     = {$(1+\epsilon, \beta)$-Spanner Constructions for General Graphs},
  journal   = {{SIAM} Journal on Comput.},
  volume    = {33},
  number    = {3},
  pages     = {608--631},
  year      = {2004},
  doi       = {10.1137/S0097539701393384},
  note      = {Announced at {STOC} 2001}
}

@inproceedings{GutenbergW20,
  author    = {Maximilian Probst Gutenberg and
               Christian Wulff{-}Nilsen},
  title     = {Fully-Dynamic All-Pairs Shortest Paths: Improved Worst-Case Time and
               Space Bounds},
  booktitle = {Proc.\ of the 2020 {ACM-SIAM} Symposium on Discrete Algorithms ({SODA} 2020)},
  pages     = {2562--2574},
  year      = {2020},
  doi       = {10.1137/1.9781611975994.156},
}

@inproceedings{AbboudW14,
  author    = {Amir Abboud and
               Virginia Vassilevska Williams},
  title     = {Popular Conjectures Imply Strong Lower Bounds for Dynamic Problems},
  booktitle = {Proc.\ of the 55th {IEEE} Annual Symposium on Foundations of Computer Science ({FOCS}
               2014)},
  year      = {2014},
  doi       = {10.1109/FOCS.2014.53}
}

@article{BernsteinBGNSSS20,
  author    = {Aaron Bernstein and
               Jan van den Brand and
               Maximilian Probst Gutenberg and
               Danupon Nanongkai and
               Thatchaphol Saranurak and
               Aaron Sidford and
               He Sun},
  title     = {Fully-Dynamic Graph Sparsifiers Against an Adaptive Adversary},
  journal   = {CoRR},
  volume    = {abs/2004.08432},
  year      = {2020},
  eprinttype = {arXiv},
  eprint    = {2004.08432}
}

@inproceedings{ForsterG19,
  author    = {Sebastian Forster and
               Gramoz Goranci},
  title     = {Dynamic low-stretch trees via dynamic low-diameter decompositions},
  booktitle = {Proc.\ of the 51st Annual {ACM} {SIGACT} Symposium on Theory
               ({STOC} 2019)},
  pages     = {377--388},
  year      = {2019},
  doi       = {10.1145/3313276.3316381}
}

@article{BernsteinFH21,
  author    = {Aaron Bernstein and
               Sebastian Forster and
               Monika Henzinger},
  title     = {A Deamortization Approach for Dynamic Spanner and Dynamic Maximal
               Matching},
  journal   = {{ACM} Transactions on Algorithms},
  volume    = {17},
  number    = {4},
  pages     = {29:1--29:51},
  year      = {2021},
  doi       = {10.1145/3469833},
  note      = {Announced at SODA 2019}
}

@inproceedings{BodwinK16,
  author    = {Greg Bodwin and
               Sebastian Krinninger},
  title     = {Fully Dynamic Spanners with Worst-Case Update Time},
  booktitle = {Proc.\ of the 24th Annual European Symposium on Algorithms ({ESA} 2016)},
  volume    = {57},
  pages     = {17:1--17:18},
  year      = {2016},
  doi       = {10.4230/LIPIcs.ESA.2016.17}
}

@article{BaswanaKS12,
  author    = {Surender Baswana and
               Sumeet Khurana and
               Soumojit Sarkar},
  title     = {Fully dynamic randomized algorithms for graph spanners},
  journal   = {{ACM} Transactions on Algorithms},
  volume    = {8},
  number    = {4},
  pages     = {35:1--35:51},
  year      = {2012},
  doi       = {10.1145/2344422.2344425}
}

@article{Elkin11,
  author    = {Michael Elkin},
  title     = {Streaming and fully dynamic centralized algorithms for constructing
               and maintaining sparse spanners},
  journal   = {{ACM} Transactions on Algorithms},
  volume    = {7},
  number    = {2},
  pages     = {20:1--20:17},
  year      = {2011},
  url       = {https://doi.org/10.1145/1921659.1921666},
  doi       = {10.1145/1921659.1921666},
  note      = {Announced at ICALP 2007}
}

@article{AusielloFI06,
  author    = {Giorgio Ausiello and
               Paolo Giulio Franciosa and
               Giuseppe F. Italiano},
  title     = {Small Stretch Spanners on Dynamic Graphs},
  journal   = {Journal of Graph Algorithms and Applications},
  volume    = {10},
  number    = {2},
  pages     = {365--385},
  year      = {2006},
  doi       = {10.7155/jgaa.00133},
  note      = {Announced at ESA 2005}
}

@article{Seidel95,
  author    = {Raimund Seidel},
  title     = {On the All-Pairs-Shortest-Path Problem in Unweighted Undirected Graphs},
  journal   = {Journal of Computer and System Sciences},
  volume    = {51},
  number    = {3},
  pages     = {400--403},
  year      = {1995},
  doi       = {10.1006/jcss.1995.1078},
  note      = {Announced at STOC 1992}
}

@inproceedings{BernsteinR11,
  author    = {Aaron Bernstein and
               Liam Roditty},
  title     = {Improved Dynamic Algorithms for Maintaining Approximate Shortest Paths
               Under Deletions},
  booktitle = {Proc.\ of the Twenty-Second Annual {ACM-SIAM} Symposium on Discrete
               Algorithms ({SODA} 2011)},
  pages     = {1355--1365},
  year      = {2011},
  doi       = {10.1137/1.9781611973082.104}
}

\appendix
\section{Randomized Algorithm for Diameter Approximation and APSP Distance Oracles}\label{app:randomized}
\label{sec:appendix:randomized}

In this section, we sketch two other implications of our techniques, namely in diameter approximation and APSP with subquadratic update time and sublinear query time, both of which were studied in \cite{BrandN19}. Unlike all of our other results these applications require randomness. 

We utilize techniques of \cite{BrandN19}, but get improved bounds by incorporating our emulator/MSSP results in their algorithms. 

\subsection{Diameter Approximation}
As discussed in Section \ref{sec:results}, we can maintain a (nearly) $(3/2+ \epsilon)$-approximation to the diameter using the algorithm of \cite{RV13}. More formally,

\begin{corollary}\label{cor:diam}
Given an unweighted graph $G=(V,E)$ with diameter $D$, and $0< \epsilon $, we can maintain an estimate $\hat{D}$ such that\footnote{Note that the term $1/3$ is only relevant for graphs with very small diameter $\leq 2$.}:
    \[ (2/3-\epsilon) D - 1/3 \le \hat{D} \le (1+\epsilon) D\]
with high probability against an adaptive adversary with the following guarantees:
\begin{itemize}[nosep]
    \item Pre-processing time of $ O(n^{\omega}) \cdot (\frac{1}{\epsilon})^{O(1)}$.
    \item $O((n^{\omega(1,1,\mu)-\mu}+n^{\omega(1, \mu, 0.5)}) (\frac{1}{\epsilon})^{O(k)} +n^{1.5+ \frac{1}{k}})$ worst-case update time for any $k$. For current $\omega$ this is $O(n^{1.596})\cdot (\frac{1}{\epsilon})^{O(1)}$.
\end{itemize}
\end{corollary}

More specifically, we use a variant used in dynamic settings by \cite{BrandN19}\footnote{In \cite{BrandNS19} use $(1+\epsilon)$-MSSP instead of exact MSSP, and obtain the same approximation guarantee. We also use $(1+\epsilon)$-MSSP in our algorithm.}.
One important component of this algorithm maintaining distances from $\tilde{O}(\sqrt{n})$ sources. That is why our improved $(1+\epsilon)$-MSSP bounds based on sparse emulators lead to better bounds for diameter approximation. For completeness we sketch the algorithm here. After each update we perform the following.

\begin{enumerate}
    \item Maintain a sparse $(1+\frac{\epsilon}{6}, \beta)$-emulator $H$ of $G$ with $\beta= (\frac{1}{\epsilon})^{k}$ and size $\tilde{O}(n^{1+\frac{1}{k}})$. 
    \item Sample a set $S \subseteq V$ of size $\tilde{O}(\sqrt{n})$ uniformly at random.

    \item Compute $(1+\epsilon)$-approximate distance estimates $\hat{d}(s,v)$ for all $s \in S, v \in V$ by (i) querying $\Theta(\beta / \epsilon)$-bounded distances in $S \times V$ using the algebraic algorithm of Lemma~\ref{lem:previous_algebraic}, (ii) computing $d_H (s,v)$ for all $s \in S, v \in V$ on the emulator~$H$ and (iii) taking the minimum of both estimates for each pair in $S \times V $. \label{alg:diam:S_sssp}

    \item Let $w \in V$ be the node with the largest distance estimate from $S$ based on the estimate computed in Step~\ref{alg:diam:S_sssp}, i.e.~ $\min_{s \in S} \hat{d}(w,s) \geq \min_{s \in S} \tilde{d}(u,s)$ for all $ u \in V $.
    \item Compute $(1+\epsilon)$-approximate distance estimates $\hat{d}(w,v)$ for all $v \in V$ by (i) querying $\Theta(\beta / \epsilon)$-bounded distances in $\{w\} \times V$ using the algebraic algorithm of Lemma~\ref{lem:previous_algebraic}, (ii) computing $d_H (w,v)$ for all $ v \in V$ on the emulator~$H$ and (iii) taking the minimum of both estimates for each pair in $\{w\} \times V $. \label{alg:diam:w_sssp}

    \item Let $W \subset V$ be the set of $\sqrt{n}$ closest nodes to $w$ based on the estimates obtained in Step~\ref{alg:diam:w_sssp}, i.e.~for all $u \in W, v \in V \setminus W$ we have $\hat{d}(w,u) \leq \hat{d}(w,v)$. Ties can be broken arbitrarily. Query $\Theta(\beta / \epsilon)$-bounded distances in $W \times V$ using the algebraic algorithm of Lemma \ref{lem:previous_algebraic}.  \label{alg:diam:W} 
    \item Compute $(1+\epsilon)$-approximate distance estimates $\hat{d}(s,v)$ for all $s \in W, v \in V$ by (i) querying $\Theta(\beta / \epsilon)$-bounded distances in $W \times V$ using the algebraic algorithm of Lemma~\ref{lem:previous_algebraic}, (ii) computing $d_H (s,v)$ for all $s \in W, v \in V$ on the emulator~$H$ and (iii) taking the minimum of both estimates for each pair in $W \times V $.

    \item Set the diameter $\hat{D}:= \max_{v \in V, u \in {S \cup W}} \{\hat{d}(u,v)\}$, i.e.~the largest estimate obtained so far.
\end{enumerate}

Note that even if we kept the set $S$ fixed, the set $W$ in the above algorithms changes in each update. Hence we need to use Lemma \ref{lem:previous_algebraic} instead of \Cref{lem:overview:pairwise}.

\begin{lemma}[{\cite{Sankowski05}}]
    \label{lem:previous_algebraic}
    There exist randomized dynamic algorithms that, after preprocessing a given unweighted directed graph $G$ and a parameter $h \in \mathbb{N}$ in $O(n^\omega h \log h)$ time, supports edge-updates to $G$ in $O((n^{\omega(1,1,\mu)-\mu}+n^{1+\mu})h \log h)$ time.
    
    The algorithm supports queries that return for any given $S,T\subseteq V$ the $S\times T$ pairwise $h$-bounded distances in $O(n^{\omega(s, \mu, t)}h\log h)$ time where $s,t$ such that $|S|=n^s$, $|T| = n^t$.\footnote{\cite{Sankowski05} considers querying entries one by one. The modification for batch-queries $S\times T$ stems from \cite{BrandNS19}.}
\end{lemma}

\begin{proof}[Proof sketch of \Cref{cor:diam}.]
The correctness (stretch) analysis follows from the following claim proved by \cite{BrandN19}: assume that we have a dynamic algorithm that returns $(1+\epsilon)$-approximate distances between pairs in $T \times V$ for \textit{any} set $T$ of size $O(\sqrt{n})$ chosen at query time. Then we can use this to compute an estimate $\hat{D}$ for the diameter satisfying Corollary \ref{cor:diam}. Using similar arguments as our MSSP algorithm in Section \ref{sec:MSSP}, the estimates~$ \hat{d} (\cdot, \cdot) $ used by the above algorithm satisfy this condition.

The update time depends on 1) the update time of the algebraic data structure for maintaining $\Theta(\beta / \epsilon)$-bounded distances from (three) sets of size $O(\sqrt{n})$,
which takes time $\tilde{O}((n^{\omega(1,1,\mu)-\mu}+n^{\omega(1, \mu, 0.5)}) \beta / \epsilon)$ via \Cref{lem:previous_algebraic}. Note that this complexity is at best $O(n^{1.5}\beta / \epsilon)$ when $\omega=2$.
2) Running MSSP statically on the emulator, which takes time $\tilde{O}(n^{\frac{3}{2}+ \frac{1}{k}})$. 

To balance out the terms, we can set $k$ to be a large enough constant and get $\beta=(\frac{1}{\epsilon})^{O(1)}$. Hence by setting $k\ge11$, we get an update time of $O((n^{\omega(1,1,\mu)-\mu}+n^{\omega(1, \mu, 0.5)}) (\frac{1}{\epsilon})^{k+1} +n^{1.5+ \frac{1}{k}}) = O (n^{1.596}) \cdot (\frac{1}{\epsilon})^{O(1)}$ .
\end{proof}

\subsection{APSP with Subquadratic Update Time and Sublinear Query Time}\label{app:subquad_APSP}

The algorithm for maintaining $(1+\epsilon)$-APSP in subquadratic update-time with sublinear query is based on an techniques used in \cite{RZ12, BrandN19}. We modify the algorithm of \cite{BrandN19} by incorporating our new dynamic emulators which allow us to improve upon the internally used dynamic $(1+\epsilon)$-MSSP algorithm.
This way we obtain the following corollary.

\begin{corollary}\label{cor:subquad_APSP}
Given an unweighted, undirected graph $G=(V,E)$, and $0<\epsilon<1$, we can maintain a data structure that supports all-pair $(1+\epsilon)$-approximate distance queries with following guarantees:
\begin{itemize}[nosep]
    \item Preprocessing time: $ O (n^{2.585}) \cdot O (\tfrac{1}{\epsilon})^{\sqrt{2\log_{1/\epsilon} n}} $.
    \item Worst-case update time: $ O (n^{1.788}) \cdot O (\tfrac{1}{\epsilon})^{\sqrt{2\log_{1/\epsilon} n}} $.
    \item Query time: $ O (n^{0.45} \epsilon^{-2}) $.
\end{itemize}
Moreover, these bounds hold with high probability against an adaptive adversary.
\end{corollary}

At a high-level the algorithm uses a well-known \textit{path hitting set}\footnote{Note that path hitting sets should not be confused with the neighborhood (of heavy nodes) hitting sets that we used in Section \ref{sec:hitting_set}.}~\cite{UllmanY91} argument as follows:
Consider any pair of nodes $u,v \in V$. If one samples $\Theta(\tfrac{n}{h} \log n)$ nodes $S\subseteq V$ uniformly at random, then w.h.p.~every shortest path of length $ h $ contains a sampled node. 
We maintain $O(\tfrac{h}{\epsilon})$-bounded distances using algebraic distance data structures.
Additionally we maintain $(1+\epsilon)$-approximate MSSP with the nodes in $ S $ as the sources using our sparse emulator.
In unweighted, undirected graphs, this information suffices to retrieve $ (1+\epsilon) $-approximate distance estimates.
Our final algorithm is a bit more involved since we need to consider two types of bounded distances: 1) We need to consider $\Theta (h/\epsilon)$-bounded distances for the parameter $h$ needed for the path hitting component of our data structure. 2) We need $\Theta (\beta/\epsilon) $-bounded distances for turning the additive approximation of our emulator into a multiplicative one. This also means we need to consider more cases in our correctness analysis.

The $\Theta(h/\epsilon)$-bounded distances will be maintained using the following dynamic algorithm.
\begin{lemma}[{\cite{BrandN19}}]\label{lem:BN_algebraic}
Given an undirected and unweighted graph $G=(V,E)$, for any $0<\tau, 0 < \epsilon, \mu_1,\mu_2 \leq 1$, we can maintain a randomized Monte-Carlo data structure against an adaptive adversary with worst-case update time of $O((n^{\omega(1,\tau+\mu_1,1)-\mu_1} \epsilon^{-1} + n^{\omega(1, \mu_2, 1)+\tau-\mu_2} + n^{1+\mu_2+\tau}) \log n)$ that can query $(1+\epsilon)$-approximate $n^\tau$-bounded pairwise distances for any set of nodes $S,T \subseteq V$ in $O(n^{\omega(\log_n |S|, \mu_1+\tau, \log_n |T|)} \epsilon^{-1} \log n)$ time.
The preprocessing time is $O(n^{\omega+\tau} \log n)$.
\end{lemma}

For current $\omega$, the best choice for minimizing the update time is roughly $\mu_2\approx0.529$, implying $O(n^{\omega(1, \mu_2, 1)+\tau-\mu_2} + n^{1+\mu_2+\tau}) = O(n^{1.529+\tau})$.

Let us continue the description of our algorithm.
The algorithm has a main parameter $ h $ and does the following.
At all times, it internally uses the following dynamic algorithms:
\begin{enumerate}
\item An instance of the dynamic algorithm of Lemma~\ref{lem:sparse_det_emulator} for maintaining a $(1 + \tfrac{\epsilon}{6}, \beta)$-emulator $H$, where $ \beta = O(\tfrac{1}{\epsilon})^k $, for $ k = \sqrt{\frac{\log_{1/\epsilon} n}{2}} $.
\item An instance $ \mathcal{D}_1 $ of the algebraic data structure of Lemma~\ref{lem:BN_algebraic} for maintaining $h_1$-bounded distances for $ h_1 = \tfrac{6 h}{\epsilon} $ (internally setting $ \tau = \log_n h_1 $), for some parameter $ h $ to be chosen later
\item An instance $ \mathcal{D}_2 $ of the algebraic data structure of Lemma~\ref{lem:previous_algebraic} for maintaining $h_2$-bounded distances for $ h_2 = \tfrac{6 \beta}{\epsilon}) $.
\end{enumerate}

Additionally, after each update to the input graph our algorithm performs the following steps:
\begin{enumerate}
    \item Sample a set of $\Theta(\tfrac{n}{h} \log n)$ sources $S$ uniformly at random, which we call centers.
    \item Compute (exact) MSSP from all nodes in $S$ on $H$ statically.
    \item For every node $ u $, find the closest center among $ S $ in $ H $, denoted by $ p(u) $.
    \item Query the $h_2$-bounded distances $ d_G^{h_2} (x, u) $ from $ \mathcal{D}_2 $ for every $ x \in S $ and every $ v \in V $  (where $ d_G^{h_2} (u, v) = d_G (u, v) $ if $ d_G^{h_2} (u, v) < h_2 $ and $ d_G^{h_2} (u, v) = \infty $ otherwise).
\end{enumerate}

Each query for the approximate distance between two nodes $ u $ and $ v $ is now answered by performing the following steps:
\begin{enumerate}
\item Query the $h_1$-bounded distance $ d_G^{h_1} (u, v) $ from $ \mathcal{D}_1 $ (where $ d_G^{h_1} (u, v) = d_G (u, v) $ if $ d_G^{h_1} (u, v) < h $ and $ d_G^{h_1} (u, v) = \infty $ otherwise).
\item Set $ \hat{d} (u, v) = \min (d^{h_2}_G (p (u), u), d_H (p(u), u)) + \min (d^{h_2}_G (p (u), v), d_H (p(u), v)) $
\item Return $ \min (d_G^{h_1} (u, v), \hat{d} (u, v)) $
\end{enumerate}

We will now show that this algorithm has the desired guarantees.


\begin{proof}[Proof of \Cref{cor:subquad_APSP}]
We first argue that the algorithm returns a $ (1 + \epsilon) $-approximation of the true distance $ d_G (u, v) $ at query time.
Clearly, the returned distance estimate never under-estimates the true distance.
If $ d_G (u, v) < h_1 $, then $ d_G^{h_1} (u, v) = d_G (u, v) $ provides the correct answer.

If $ d_G (u, v) \geq h_1 $, then consider the first $ \tfrac{\epsilon}{6} h_1 = h $ nodes on the shortest path from $ u $ to $ v $.
This set of nodes contains a center $ x \in S $ with high probability.
We thus know that
\begin{equation*}
    d_G (u, p(u)) \leq d_G (u, x) \leq \frac{\epsilon}{6} h_1 \leq \frac{\epsilon}{6} d_G (u, v) \, .
\end{equation*}
By the triangle inequality we now get
\begin{align*}
d_G (p(u), u) + d_G (p(u), v) &\leq d_G (p(u), u) + d_G (p(u), u) + d_G (u, v) \\
 &= d_G (u, v) + 2 d_G (p(u), u) \\
 &\leq \left(1 + \frac{\epsilon}{3}\right) d_G (u, v) \, .
\end{align*}

By the stretch guarantee of the emulator $ H $ we have $ d_H (p(u), u) \leq (1 + \tfrac{\epsilon}{6}) d_G (p(u), u) + \beta $ and $ d_H (p(u), v) \leq (1 + \tfrac{\epsilon}{6}) d_G (p(u), v) + \beta $.
If $ d_G (p(u), u) \geq \tfrac{6 \beta}{\epsilon} $, then $ d_H (p(u), u) \leq (1 + \tfrac{\epsilon}{3}) d_G (p(u), u) $, and if $ d_G (p(u), u) < \tfrac{6 \beta}{\epsilon} = h_2 $, then $ d^{h_2}_G (p (u), u) = d_G (p (u), u) $.

Thus, $ \min (d^{h_2}_G (p (u), u), d_H (p(u), u)) \leq (1 + \tfrac{\epsilon}{3}) d_G (p(u), u) $.
We can argue in the same manner that $ \min (d^{h_2}_G (p (u), v), d_H (p(u), v)) \leq (1 + \tfrac{\epsilon}{3}) d_G (p(u), v) $.
Overall, this gives us
\begin{align*}
\hat{d} (u, v) &\leq \left(1 + \frac{\epsilon}{3}\right) d_G (p(u), u) + \left(1 + \frac{\epsilon}{3}\right) d_G (p(u), v) \\
 &= \left(1 + \frac{\epsilon}{3}\right) (d_G (p(u), u) + d_G (p(u), v)) \\
 &\leq \left(1 + \frac{\epsilon}{3}\right)^2 d_G (u, v) \\
 &\leq (1 + \epsilon) d_G (u, v) \, .
\end{align*}

We proceed with the running time analysis.
Data structure $ \mathcal{D}_1 $ has, for any chosen $ \mu_ 1 $, an update time of $O((n^{\omega(1,\mu_1 + \log_n h_1,1)-\mu_1} \epsilon^{-1} + O(n^{1.529 + \log_n h_1}) \log n) = \tilde O (n^{\omega(1,\mu_1 + \log_n h,1)-\mu_1} \epsilon^{-2} + O(n^{1.529 + \log_n h} \epsilon^{-1})$.\footnote{Here we use that $ \omega (\alpha, \beta + \beta', \gamma) = O (\omega (\alpha, \beta, \gamma) n^{\beta'}) $. We will also use this bound later in the proof.}
Data structure $ \mathcal{D}_2 $ has, for any chosen $ \mu_2 $, an update time of $ O ((n^{\omega(1,1,\mu_2)-\mu_2} + n^{1+\mu_2}) h_2 \log h_2) = \tilde O ((n^{\omega(1,1,\mu_2)-\mu_2} + n^{1+\mu_2})) \cdot O(\tfrac{1}{\epsilon})^k $.\footnote{Note that these simplifications assume that $ \epsilon \geq \tfrac{1}{n} $ as a smaller value of $ \epsilon $ would allow for rounding to exact distances. We will also use this bound later in the proof.}

The update time of the algorithm maintaining the emulator $ H $ is $ O (n^{1.477}) $ (which is clearly dominated by the update time of $ \mathcal{D}_1 $) as in any case we can reduce $ \epsilon $ to a value that gives $ k \geq 8 $.
Sampling the set $ S $ takes time $ \tilde O (n) $.

Computing MSSP from $ S $ on $ H $ takes time $ \tilde O (|S| n^{1 + 1/k}) = \tilde O (\tfrac{n^{2 + 1/k}}{h}) $.
Computing the closest node among $ S $ for each node takes time $ O (n |S|) = \tilde O (\tfrac{n^2}{h}) $.

The time for querying the $ h_2 $-bounded $ S \times V $ distances from data structure $ \mathcal{D}_2 $ is bounded by $ O (n^{\omega (\log_n |S|, \mu_2, 1)} h_2 \log h_2) = \tilde O (n^{\omega (1 - \log_n h, \mu, 1)}) \cdot O(\tfrac{1}{\epsilon})^k $.
The running time of the remaining steps performed during an update is dominated by $ \tilde O (|S| n^{1 + 1/k}) = \tilde O (\tfrac{n^{2 + 1/k}}{h}) $.

In the query algorithm, we first query data structure $ \mathcal{D}_1 $ for one pair of nodes, which takes time $O(n^{\omega(0, \mu_1 + \log_n h_1, 0)} \epsilon^{-1} \log n) = \tilde O (n^{\mu_1} h \epsilon^{-2}) $.
The remaining steps in the query algorithm are clearly dominated by this term.

The preprocessing time of $ \mathcal{D}_1 $ is $ O (n^{\omega + \log_n h_1} \log n) = \tilde O (n^{\omega} h \epsilon^{-1}) $ and the preprocessing time of~$ \mathcal{D}_2 $ is $ O (n^{\omega} h_2 \log h_2) = \tilde O (n^\omega) \cdot O(\tfrac{1}{\epsilon})^k $.

Our approach for setting the parameters is to consider $ n^{\omega(1,\mu_1 + \log_n h,1)-\mu_1} $ and $ \tfrac{n^2}{h} $ as the dominant terms and set $ h $ and $ \mu_1 $ in such a way that $ n^{\mu_1 + \log_n h} = n^{0.45} $.
The latter is equivalent to $ \mu_1 + \log_n h = 0.45 $, which yields the simplification $ n^{\omega(1,\mu_1 + \log_n h,1)-\mu_1} = n^{\omega(1, 0.45 ,1)-\mu_1} $.
We now balance the terms $ n^{\omega(1, 0.45 ,1)-\mu_1} $ and $ \tfrac{n^2}{h} $ under the constraint $ n^{\mu_1 + \log_n h} < n^{0.45} $.

By setting $ \mu_1 = 0.2374 $ and $ h = n^{0.2125} $ we get $ \tilde O (n^{\omega(1,\mu_1 + \log_n h,1)-\mu_1}) = n^{1.788} $, $ \tilde O (\tfrac{n^2}{h}) = n^{1.788} $, and $ \tilde O (n^{\mu_1} h) = O (n^{0.45}) $ . 
Now by setting $ \mu_2 = 0.24 $ we get $ \tilde O ((n^{\omega(1,1,\mu_2)-\mu_2} + n^{1+\mu_2})) = O (n^{1.76}) $ and $ \tilde O (n^{\omega (1 - \log_n h, \mu, 1)}) =  n^{1.788} $.
Finally, setting $ k = \sqrt{\frac{\log_{1/\epsilon} n}{2}} $ balances $ (\tfrac{1}{\epsilon})^k $ and $ n^{1/k} $ in the update time.

Overall, we obtain an update time of $ O (n^{1.788}) \cdot O (\tfrac{1}{\epsilon})^{\sqrt{2\log_{1/\epsilon} n}} $, a query time of $ O (n^{0.45} \epsilon^{-2}) $, and a preprocessing time of $ O (n^{2.585}) \cdot O (\tfrac{1}{\epsilon})^{\sqrt{2\log_{1/\epsilon} n}} $.
\end{proof}

\section{Comparison to Algebraic Algorithms with Batch-Queries}
\label{sec:batchquerycomparison}

In \Cref{sec:overview:emulator}
we remarked that the emulator construction with low recourse hitting sets is needed for our upper bounds to match conditional lower bounds from \cite{BrandNS19}.
Here we state the complexity that could be achieved if one were to use a new hitting set in each iteration instead.

The bottleneck is maintaining the submatrix $\mM^{-1}_{S,T}$ of some dynamic matrix inverse. \Cref{lem:algebraic:submatrix} is able to maintain such a submatrix efficiently, if sets $S$ and $T$ are slowly changing.
If the sets, however, were changing arbitrarily from one iteration to the next (e.g.~by using a new hitting set in our emulator construction after each update), then we would have to use previous algorithms \cite{Sankowski04,BrandNS19} instead with the following complexities.
\begin{lemma}[\cite{Sankowski04}]\label{lem:appendix:sank}
\label{lem:overview:previous_algebraic}
\label{lem:algebraic:slowUpdate}
For any $0\le\nu\le1$ there exists a dynamic matrix inverse algorithm that initializes in $O(n^\omega)$ operations and supports entry updates to $\mA$ in $O(n^{\omega(1,1,\nu)-\nu}+n^{1+\nu})$ operations.
Querying $\mA^{-1}_{S,T}$ for any $|S|=n^\sigma$, $|T|=n^\tau$ takes
$O(n^{\omega(\sigma,\nu,\tau)})$ operations.

\end{lemma}
In comparison, our \Cref{lem:algebraic:submatrix} (when picking $\mu=1$) replaces the $O(n^{\omega(\sigma,\nu,\tau)})$ query complexity of \Cref{lem:appendix:sank} by a smaller additive $O(|S\times T|)=O(n^{\omega(\sigma, 0,\tau)})$ in the update complexity.

For the emulator construction for $(1+\epsilon)$-SSSP we pick $|S|=n$ and $T=\tilde{O}(\sqrt{n})$.
So for current $\omega$ and constant $\epsilon$, \Cref{lem:appendix:sank}
would only imply $O(n^{1.596})$ for single source distances ($\sigma=1,\tau=0.5,\nu\approx0.42$)
opposed to our deterministic $O(n^{1.529})$ upper bound
(which was also achieved by \cite{BHGWW2021} with randomization against oblivious adversaries).

For the emulators used in the $st$-case one could use the following lemma.
\begin{lemma}[\cite{BrandNS19}]\label{lem:appendix:bns}
For any $0\le\nu\le\mu\le1$ there exists a dynamic matrix inverse algorithm that initializes in $O(n^\omega)$ operations and supports entry updates to $\mA$ in $O(n^{\omega(1,1,\mu)-\mu}+n^{\omega(1,\mu,\nu)-\nu}+n^{\mu+\nu})$ operations.
Querying $\mA^{-1}_{S,T}$ for any $|S|=n^\sigma$, $|T|=n^\tau$ takes
$O(n^{\omega(\sigma,\mu,\nu)} + n^{\omega(\tau,\mu,\nu)}
+n^{\omega(\sigma,\nu,\tau)})$ operations.
\end{lemma}
In comparison, our \Cref{lem:algebraic:submatrix} replaces all terms in the query complexity of \Cref{lem:appendix:bns} by a much smaller additive $O(|S\times T|)=O(n^{\omega(\sigma, 0,\tau)})$ in the update complexity.

We use $|S|=|T|=\tilde{O}(n^{2/3})$ in our emulator construction for $(1+\epsilon)$-approximate $st$-distances.
For current $\omega$ and constant $\epsilon$, \Cref{lem:appendix:bns} would imply $O(n^{1.438})$ update time for $st$-distances ($\mu\approx0.747,\nu\approx0.323,\sigma=\tau=2/3$)
opposed to our $O(n^{1.407})$ bound.

\section{Improved Randomized Bounds for Exact $st$-Distances}\label{app:exact_st}

We remark that our new dynamic algorithm for maintaining bounded distances (\Cref{thm:randomized:low_hop}) can also be used to speed up dynamic exact $st$-distances in directed graphs, if we allow for randomization.
The algorithm is randomized and works on directed graphs.
We obtain an update time of $O(n^{1.7035})$, 
improving upon the previous best $O(n^{1.7643})$ \cite{Sankowski05,BrandNS19}.

\exactst*

\begin{proof}
Given graph $G$, we sample a random hitting set $H$ of size $\tilde O(n/h)$ and add $s,t$ to this set. 
With high probability, this set partitions the shortest $st$-path into paths $v \to v'$ (with $v,v'\in H$) of length $\le h$.
We maintain the pairwise $h$-bounded distances for $H\times H$ via the data structure of \Cref{thm:randomized:low_hop}.
After each update, we construct a graph $G' = (H, E')$ on vertex set $H$ with edges whose weight matches the maintained $h$-bounded distances.
The $st$-distance in $G$ can now be computed by running Dijkstra's algorithm on $G'$ in $\tilde{O}((n/h)^2)$ time.
This complexity is subsumed by the time required by \Cref{thm:randomized:low_hop} to maintain the $H\times H$ distances (i.e.~the weights to be used to construct $G'$.

As the result is w.h.p.~correct and exact, no information about the random choices is leaked to the adversary.
\end{proof}

\end{document}